\documentclass{l4dc2025}

\usepackage{wrapfig}
\usepackage{hyperref}
\usepackage{cleveref}
\usepackage{algorithm}
\usepackage{algorithmicx}
\usepackage[noend]{algpseudocode}
\usepackage{graphicx}
\usepackage{Shorthands}
\usepackage{enumitem}
\usepackage{ifthen}
\usepackage{floatflt}
\usepackage[font=small,labelfont=bf]{caption}
\usepackage{array}
\usepackage{amsmath}
\usepackage{booktabs}
\renewcommand{\arraystretch}{1.6}

\newcounter{l4dc}

\makeatletter
\renewcommand{\jmlrvolume}[1]{}
\renewcommand{\jmlryear}[1]{}
\renewcommand{\jmlrproceedings}[2]{}
\def\ps@jmlrtps{%
  \def\@oddhead{}
  \def\@evenhead{}
  \def\@oddfoot{}
  \def\@evenfoot{}
}
\makeatother

\usepackage[noindentafter]{titlesec}
\titlespacing\section{0pt}{10pt}{4pt}
\titlespacing\subsection{0pt}{4pt}{3pt}

\usepackage{times}
\title[]{Domain Randomization is Sample Efficient\\ for Linear Quadratic Control}

\author{%
 \Name{Tesshu Fujinami} \Email{ftesshu@seas.upenn.edu} \\
 \Name{Bruce D. Lee} \Email{brucele@seas.upenn.edu} \\
 \Name{Nikolai Matni} \Email{nmatni@seas.upenn.edu} \\
 \Name{George J. Pappas} \Email{pappasg@seas.upenn.edu} \\
 \addr All authors are with the Department of Electrical and Systems Engineering, University of Pennsylvania}

\date{August 2023}
\begin{document}
\maketitle

\begin{abstract} 
   We study the sample efficiency of domain randomization and robust control for the benchmark problem of learning the linear quadratic regulator (LQR). Domain randomization, which synthesizes controllers by minimizing average performance over a distribution of model parameters, has achieved empirical success in robotics, but its theoretical properties remain poorly understood. We establish that with an appropriately chosen sampling distribution, domain randomization achieves the optimal asymptotic rate of decay in the excess cost, matching certainty equivalence. We further demonstrate that robust control, while potentially overly conservative, exhibits superior performance in the low-data regime due to its ability to stabilize uncertain systems with coarse parameter estimates. We propose a gradient-based algorithm for domain randomization that performs well in numerical experiments, which enables us to validate the trends predicted by our analysis. These results provide insights into the use of domain randomization in learning-enabled control, and highlight several open questions about its application to broader classes of systems. 
\end{abstract}

\section{Introduction}

\begin{wrapfigure}[12]{r}{0.55\textwidth}
    \centering
    \vspace{-10pt} 
    \includegraphics[width=\linewidth]{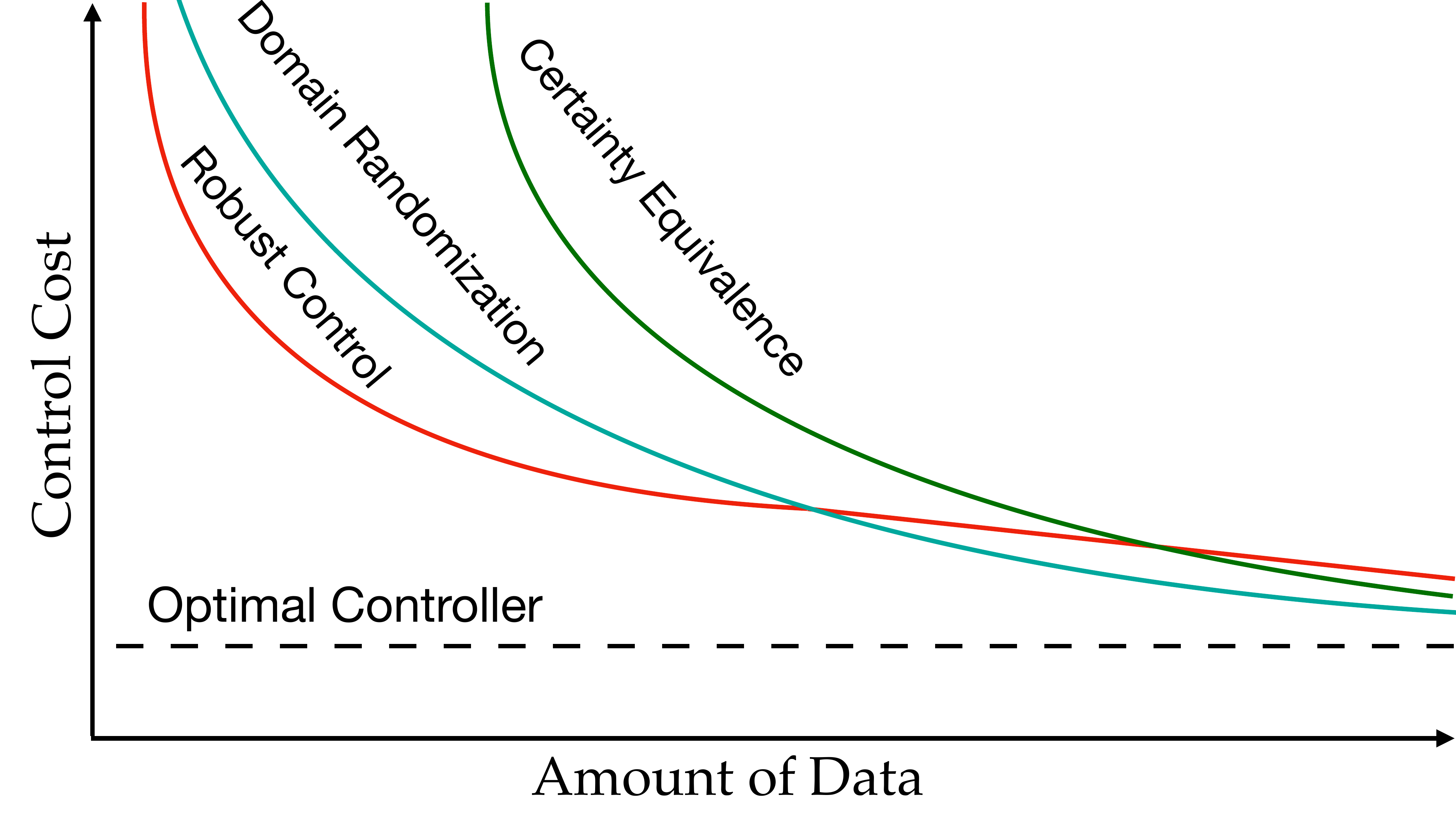}
    \vspace{-24pt} 
    \caption{Illustration of the sample efficinecy of various synthesis methods. }
    \label{fig:conceptual figure}
\end{wrapfigure}

The use of learned world models to synthesize controllers via policy optimization is becoming increasingly prevalent in reinforcement learning \citep{wu2023daydreamer, matsuo2022deep}. The performance of the resulting controller depends heavily upon the synthesis procedure. Simple approaches that do not account for uncertainty, known as certainty equivalence, can overfit to errors in the learned model. Robust control approaches can tolerate some error in the learned model, but may be overly conservative and computationally demanding. 
Consequently, \emph{domain randomization} has emerged as a dominant paradigm in robotics for enabling transfer of policies optimized in simulation on a learned or physics-based simulator to the real world by randomizing system parameters during policy optimization \citep{tobin2017dr, akkaya2019solving}. Despite the empirical success of domain randomization, it remains poorly understood how to select the randomization distribution, and how much of a discrepancy between the learned model and the real world can be tolerated. 

We seek to address these issues by restricting attention to the benchmark problem of learning the linear quadratic regulator (LQR). This problem consists of collecting experimental interaction data from a linear dynamical system, and using this data to synthesize a controller that optimizes a control objective. The linear dynamical system is described by
\begin{align}
    \label{eq: linear system}
    X_{t+1} = A(\theta^\star) X_t + B(\theta^\star) U_t + W_t \textrm{ for } t=1,\dots, T, 
\end{align}
where $X_t\in\R^{\dx}$ is the system state, $U_t\in\R^{\du}$ is the control input, $W_t\in\R^{\dx}$ is i.i.d. mean zero Gaussian noise, and $\theta^\star\in\R^{d_{\theta}}$ is an unknown parameter. The goal is to design a controller $K$ to minimize the objective $C(K, \theta^\star)$, 
defined as
\begin{align}
    \label{eq: lqr cost}
    C(K, \theta) \triangleq \limsup_{T\to\infty}\E_{\theta}^K 
    \brac{ \frac{1}{T} \sum_{t=1}^T \paren{\norm{X_t}_Q^2 + \norm{U_t}_R^2}},
\end{align}
for $Q$ and $R$ positive definite weight matrices. The subscript on the expectation denotes that the states evolves according to $X_{t+1} = A(\theta) X_t +B(\theta) U_t + W_t$, and the superscript denotes that the inputs are selected according to the linear feedback $U_t = K X_t$. This benchmark problem has been used to study the sample efficiency of certainty equivalence \citep{mania2019certainty} and robust control \citep{dean2020sample} by quantifying how many experiments from the system \eqref{eq: linear system} are sufficient to achieve some level of control performance. In this work, we study the sample efficiency of domain randomization, which chooses a control policy as 
\begin{align}
    \label{eq: domain randomization}
    K_{\mathsf{DR}}(\calD) = \argmin_{K} \E_{\theta\sim \calD} \brac{C(K, \theta)}, 
\end{align}
for a sampling distribution $\calD$ determined using the dataset of experiments collected from \eqref{eq: linear system}.

\subsection{Contributions}

Our contribution is to study the sample efficiency of domain randomization and robust control to establish the relationships visualized as a conceptual diagram in \Cref{fig:conceptual figure}. In particular, we achieve the following: 
\begin{itemize}[noitemsep, nolistsep, leftmargin=*]
    \item \textbf{Sample Effiency of Domain Randomization:} We prove that with an appropriately chosen sampling distribution $\calD$, the domain randomization procedure of \eqref{eq: domain randomization} achieves the optimal asymptotic rate of decay with the number of samples, thereby matching the performance of certainty equivalence in the large sample regime. We further conjecture that the burn-in time for DR lies between that of RC and CE.  We leave proving this to future work, but verify this conjecture numerically.
    \item \textbf{Sample Effiency of Robust Control:} We prove the tightest known bound on robust control, improving the asymptotic rate of decay with the number of samples $N$ from $1/\sqrt{N}$ to $1/N$. The upper bounds indicate a gap between the asymptotic rate of decay for robust control, and the rate of decay for domain randomization and certainty equivalence in terms of system-theoretic quantities. We conjecture that this gap is fundamental, due to the conservative nature of robust control. However, we establish that robust control can achieve a smaller burn-in time relative to certainty equivalence, due to its ability to stabilize the system with a coarse estimate. 
    \item \textbf{Algorithm for Domain Randomization:} We propose an algorithm to solve \eqref{eq: domain randomization} which proves effective in numerical experiments. This enables verification of the trends predicted in the aforementioned results, and aligns with the conceptual diagram in \Cref{fig:conceptual figure}. While the focus of this work is restricted to linear systems, the proposed algorithm can in principle extend to general nonlinear systems, whereas similar extensions for robust control face computational challenges.  
\end{itemize}
By providing this characterization, our work demonstrates the potential of domain randomization in learning-enabled control, and partially explains the empirical success that it has achieved in robotics applications. We therefore conclude by highlighting several interesting open questions regarding the use of domain randomization for learning-enabled control. 





\subsection{Related Work}

\paragraph{Domain Randomization} 
Domain randomization, introduced by \citet{tobin2017dr}, is widely used for \emph{sim-to-real transfer}. By randomizing simulator parameters during training, it aims to produce policies robust to simulator variations, thereby enabling transfer to the real-world. This approach has been applied in areas like autonomous racing \citep{loquercio2019deep}
and robotic control \citep{peng2018drforcontrol, akkaya2019solving}. However, its success depends heavily on selecting an effective sampling strategy \citep{mehta2020active}, which is often challenging. While previous work has explored generalization of domain randomization in discrete Markov Decision Processes \citep{chen2021understanding, zhong2019pacreinforcementlearningrealworld, jiang2018pac}, formalizing generalization for continuous control remains an open problem, which we address in this work.
\vspace{-3pt}
\paragraph{Identification and Control} The linear quadratic regulator problem has become a key benchmark for evaluating reinforcement learning in continuous control \citep{abbasi2011regret, recht2019tour}. The offline setting has been extensively studied: \citet{dean2020sample} analyzed the sample efficiency of robust control, while \citet{mania2019certainty, wagenmaker2021task, lee2023fundamental} showed that certainty equivalence is asymptotically instance-optimal, achieving the best possible sample efficiency with respect to system-theoretic quantities. Extensions to smooth nonlinear systems were made by \citet{wagenmaker2024optimal, lee2024active}. However, certainty equivalence can perform poorly with limited data. Alternative Bayesian approaches \citep{von2022improving, chiuso2023harnessing} can mitigate such limitations. We therefore show that such uncertainty-aware synthesis methods can match the asymptotic efficiency of certainty equivalence while achieving better performance in low-data regimes.
\vspace{-3pt}
\paragraph{Robust Control:} The control community has traditionally addressed policy synthesis with imperfect models using methods like $\calH_\infty$ control, which focuses on worst-case uncertainty \citep{zhou1996robust, bacsar2008h, doyle1982analysis, fan1991robustness}. Randomized approaches to robust control emphasizing high-probability guarantees have also been explored \citep{calafiore2006scenario, stengel1991technical, ray1993monte}. \citet{vidyasagar2001randomized} proposed an average performance metric similar to domain randomization but focused on a fixed distribution rather than one informed by data. Early data-driven synthesis efforts combined classical system identification \citep{ljung1998system} with worst-case robust control \citep{gevers2005identification}, while recent work has developed robust synthesis methods that bypass explicit models \citep{berberich2020robust}. To the best of our knowledge, existing analyses of statistical efficiency in robust control yield suboptimal rates, with excess control cost decreasing at $1/\sqrt{N}$ \citep{dean2020sample}, compared to the faster $1/N$ rate achieved by certainty equivalence. This work refines robust synthesis analysis, demonstrating the $1/N$ rate and a short burn-in period, highlighting its advantages with limited data.

\vspace{-3pt}
\paragraph{Notation: }
 \sloppy
The operator $\mathsf{vec}(A)$ stacks the columns of $A$ into a vector, and its inverse $\mathsf{vec}^{-1}(\mathsf{vec}(A), n)$ returns a matrix with $n$ rows by successively stacking chunks of $n$ elements into columns. The Kronecker product of $A$ and $B$ is $A \kron B$. $x \vee y$ denotes the max of $x$ and $y$.

\section{Problem Formulation}

Consider the linear dynamical system \eqref{eq: linear system}. We assume that$(A(\theta^\star), B(\theta^\star))$ is stabilizable. For ease of exposition, we restrict attention to the case where $\bmat{A(\theta) & B(\theta)} = \mathsf{vec}^{-1}(\theta, \dx)$, i.e. all entries of the state and input matrices are unknown. We additionally assume that $\Sigma_w = I$, and the cost matrices $Q\succeq I$ and $R = I$ are known.\footnote{Generalizing to arbitrary $\Sigma_w \succ 0$, $Q\succ 0$, and $R\succ 0$ can be performed by scaling the cost and changing the state and input basis. }

When $\theta$ is known, the optimal controller that minimizes $C(K, \theta)$ is given by 
\begin{align*}
    P(\theta) &\triangleq A(\theta)^\top P(\theta) A(\theta) - A(\theta)^\top P(\theta) B(\theta) (B(\theta)^\top P(\theta)  B(\theta)+ R)^{-1} B(\theta)^\top P(\theta) A(\theta)+ Q \\
    K(\theta) &\triangleq -(B(\theta)^\top P(\theta) B(\theta) + R)^{-1} B(\theta)^\top P(\theta) A(\theta),
\end{align*}
where $P(\theta)$ is the positive definite solution to the discrete algebraic Ricatti equation, and $K(\theta)$ is the LQR solution corresponding to a system with parameters $\theta$. 

To design a controller for the unknown system \eqref{eq: linear system}, we suppose that we have run $N$ experiments with control input $U_t \sim \calN(0, \Sigma_u)$ for $\Sigma_u \succ 0$.  From these experiments, we collect a dataset 
\begin{align}
    \label{eq: dataset}
    \curly{(X_t^n, U_t^n, X_{t+1}^n)}_{t=1, n=1}^{T,N}
\end{align} consisting of $N$ trajectories of length $T$ from \eqref{eq: linear system}. We use this dataset to design a controller $\hat K$ such that the cost, $C(\hat K, \theta^\star)$,
is small. In our analysis, we consider several approaches to achieve this objective. The approaches will be  contrasted by examining the rate at which this cost decays to the optimal cost as the number of experiments in the dataset increases.
\subsection{Controller Synthesis Approaches}
\label{s: methods}

All the synthesis approaches under consideration are model-based approaches which determine $\hat \theta$ from the dataset \eqref{eq: dataset} via the following least squares problem:
\begin{align}
    \label{eq: least squares}
    \hat \theta \triangleq \argmin_{\theta} \sum_{n=1}^N \sum_{t=1}^T \norm{X_{t+1}^n - \bmat{A(\theta) & B(\theta)} \bmat{X_t \\ U_t}}^2. 
\end{align}
The estimation error for $\hat \theta$ can be characterized using the Fisher information matrix:
\begin{align}
    \mathsf{FI}(\theta^\star) \triangleq \mathbf{E}_{\theta^\star}^{U_t \sim \calN(0, \Sigma_u)} \brac{\sum_{t=1}^T \bmat{X_t \\ U_t} \bmat{X_t \\ U_t}^\top } \otimes \Sigma_W^{-1}, \label{eq:FI}
\end{align}
see, e.g. \citet{lee2024active}. Since this quantity depends on the unknown parameter, we define the estimate
\begin{align}
     \label{eq: fisher estimate}
    \hat{\mathsf{FI}} \triangleq \frac{1}{N}\sum_{n=1}^N \sum_{t=1}^T \bmat{X_t^n \\ U_t^n } \bmat{X_t^n \\ U_t^n}^\top \otimes \Sigma_w^{-1},
\end{align}
which quantifies the uncertainty of the estimation procedure. 

We consider three approaches to control synthesis using the estimates $\hat\theta$ and $\hat{\mathsf{FI}}$:
\begin{itemize}[noitemsep, nolistsep, leftmargin=*]
    \item \textbf{Certainty Equivalence (CE)} uses the estimate $\hat{\theta}$ to minimize the control objective (\ref{eq: lqr cost}) by treating the estimate as though it were ground truth: 
        $K_{\mathsf{CE}}(\hat \theta) = \argmin_{K} C(K,\hat\theta).$
    \item \textbf{Robust Control (RC)} constructs a high confidence ellipsoid around the nominal estimate $\hat\theta$ using the estimated fisher information matrix $\hat{\mathsf{FI}}$ as
    \begin{align}
        \label{eq: confidence ellipsoid}
        G = \curly{\theta: (\theta-\hat\theta)^\top (N\hat {\mathsf{FI}}) (\theta- \hat \theta) \leq 16 (d_{\theta} +  \log(2/\delta))}.
    \end{align}
    Such a set can be shown to contain the true parameter $\theta^\star$ with probability at least $1-\delta$ as long as the number of experiments, $N$, is sufficiently large (see \Cref{s: id bound proof}). 
    Robust control then uses the confidence ellipsoid to determine a controller that minimizes the worst case  value of the control objective over all members of the confidence set as
    \begin{align*}
        {K}_{\mathsf{RC}}(G) = \argmin_{K}\sup_{\theta\in G}(C(K, {\theta}) - C(K(\theta), {\theta})).
    \end{align*}
    This formulation is nonstandard, as the controller minimizes the worst-case suboptimality gap, rather than the worst case cost \citep{gevers2005identification}. However, it simplifies the analysis. 
    \item \textbf{Domain Randomization (DR)} constructs a sampling distribution $\calD$ using the least squares estimate $\hat \theta$ and the estimated Fisher Information $\hat{\mathsf{FI}}$. It then synthesizes a controller by minimizing the average control cost as in \eqref{eq: domain randomization}. By ensuring good performance on average over a sampling distribution, domain randomization serves as a middle ground between certainty equivalence and robust control. In particular, it can be interpreted as enforcing a high probability robust stability constraint over the sampling region.\footnote{If the controller is not stabilizing for a subset of systems with nonzero mass, the cost will be infinite.} Careful choice of the sampling distribution is therefore critical for downstream performance. Our analysis informs this choice. 
\end{itemize}

The goal of this paper is to study the sample efficiency of these three approaches. In particular, we consider upper bounds on the gap $C(\hat K, \theta^\star)- C(K(\theta^\star), \theta^\star)$, where $\hat K$ is a controller synthesized with CE, RC, or DR. We express these bounds in terms of system-theoretic quantities, and the number of experiments collected from the system. Doing so provides an indication of the types of systems on which these methods perform well. We focus our attention on two key quantities: the burn-in time required to ensure finite bounds, and the asymptotic rate of decay in these bounds.

\section{Sample Efficiency Bounds for Controller Synthesis Approaches}

\vspace{-5pt}

\begin{table}[h!]
    \centering
    \begin{tabular}{|c|c|c|c|c|}
    \hline
        Method\! & Leading Term & Burn-in Time & \!Scalable Alg.\! & Source of Bounds \\
    \hline
     CE & $\frac{1}{N}\trace\paren{H(\theta^\star) \mathsf{FI}(\theta^\star)^{-1}}$ & $ {\norm{P(\theta^\star)}^{10}}$ & Yes & \!\citet{wagenmaker2021task}\footnotemark\! \\
     \hline
     RC & $\frac{1}{N} d_{\theta} \norm{H(\theta^\star) \mathsf{FI}(\theta^\star)^{-1}}$ & ${\!\norm{P(\theta^\star)}^4\!\vee\! \frac{1}{r^2}}\!$ & No & This paper \\
     \hline
     DR & $\frac{1}{N}\trace\paren{H(\theta^\star) \mathsf{FI}(\theta^\star)^{-1}}$ & $ \!{\norm{P(\theta^\star)}^{11}} \tau_{B(\theta^\star\!)}^{16}\!$ & Yes & This paper \\
     \hline
    \end{tabular}
    \vspace{-6pt}
    \caption{Comparison of sample efficiency bounds for CE, RC, and DR. The leading term is stated up to universal constants. The burn-in time reports the components which are different between the three approaches. We use the shorthand $\tau_{B(\theta^\star)} = \norm{B(\theta^\star)} \vee 1$. We classify algorithms as scalable if they are possible to implement via first order gradient-based approaches.}
    \vspace{-2pt}
    \label{tab:sample_efficiency}
\end{table}
\footnotetext{Due to slight discrepancies in the setting (e.g. we consider multiple trajectories for identification), the exact version of the certainty equivalent bound considered is presented in \Cref{s: certainty equivalence bound}.}

Our sample efficiency bounds are summarized in \Cref{tab:sample_efficiency}, where they are compared with existing bounds for certainty equivalence. The leading term in the bound characterizes the asymptotic rate of decay. For all three synthesis approaches described in \Cref{s: methods}, the leading term depends on four quantities: the parameter dimension, $d_{\theta}$, the number of experiments, $N$, the Fisher Information matrix $\mathsf{FI}(\theta^\star)$, and a matrix $H(\theta^\star)$, which captures the sensitivity of control synthesis to the estimation error. In particular, $H(\theta^\star)$ is given by
\begin{align*}
    H(\theta^\star) = \nabla_{\theta}^2 C(K(\theta), \theta^\star)\vert_{\theta=\theta^\star}.
\end{align*}
\citet{wagenmaker2021task} show that $\frac{1}{N}\trace\paren{H(\theta^\star) \mathsf{FI}(\theta^\star)^{-1}}$ is the optimal asymptotic rate achievable by any algorithm mapping a dataset \eqref{eq: dataset} to a controller.\footnote{This lower bound is specified to our setting in \Cref{s: lower bound proof}.} Accordingly, both certainty equivalence and domain randomization can be classified as sample-efficient, as they achieve this optimal rate of decay with respect to system-theoretic quantities. The bound on robust control instead has a leading term of $d_{\theta} \norm{H(\theta^\star) \mathsf{FI}(\theta^\star)^{-1}} \geq \trace\paren{H(\theta^\star) \mathsf{FI}(\theta^\star)^{-1}}$, and therefore cannot be classified as efficient.\footnote{While we lack a formal lower bound proving the inefficiency of RC, experiments support this conclusion (see Fig.~\ref{fig:result dr lqr}).} 

\Cref{tab:sample_efficiency} also highlights the burn-in time, the number of samples that suffice for the sample efficiency bounds to hold. The reported values omit terms common to all three methods. Among the differing quantities, the burn-in for CE scales with $\norm{P(\theta^\star)}^{10}$, for DR it scales with $\norm{P(\theta^\star)}^{11}\tau_{B(\theta^\star)}^{12}$, and for RC it scales with the max of $\norm{P(\theta^\star)}^{4}$ and $\frac{1}{r^2}$, a term quantifying the robust stabilizability of $\theta^\star$. Although $\frac{1}{r^2}$ can be as large as $\norm{P(\theta^\star)}^{10}$, it is often much smaller (see \Cref{s: robust control}), suggesting that robust control can achieve a much lower burn-in than the alternative approaches for many system instances.\footnote{We emphasize that these burn-in conditions are sufficient but not necessary; refining them is left to future work.} Experiments further suggest that domain randomization's burn-in can potentially fall between that of certainty equivalence and robust control.

Finally, \Cref{tab:sample_efficiency} highlights that certainty equivalence and domain randomization give rise to scalable gradient-based policy optimization algorithms, which can easily be extended to nonlinear and high dimensional systems. In contrast, the solving the robust control problem requires computationally challenging LMI-based approaches, even for fully observed linear systems.

\subsection{Sample Efficiency of Domain Randomization}

We now establish the characterization of domain randomization in the final row of \Cref{tab:sample_efficiency}. To this end, we first define a burn-in time which enables a bound on the the least squares  error: 
\begin{align}
    \label{eq: id burn-in}
    N_{\mathsf{ID}} \triangleq \mathsf{poly}\paren{\sum_{t=0}^{T-1} \norm{A(\theta^\star)^t \bmat{I & B(\theta^\star)} }, \norm{\Sigma_u}, \norm{\Sigma_w}, \norm{\mathsf{FI}(\theta^\star)}, \frac{1}{\lambda_{\min}\paren{\mathsf{FI}(\theta^\star)}} }.
\end{align}
A bound on least squares using this quantity is shown in \Cref{s: id bound proof}. Proofs for the remainder of the results in this section may be found in \Cref{s: domain randomization proofs}. 

We first state a bound on the performance of domain randomization that holds for general sampling distributions centered at $\hat \theta$. 
\begin{lemma}
    \label{lem: domain randomization general}
    Suppose the dataset $\curly{(X_t^n, U_t^n, X_{t+1}^n)}_{t=1, n=1}^{T,N}$ is collected from N trajectories of the system \eqref{eq: linear system} via a random control input $U_t \sim \calN(0, \Sigma_u)$. Let $\hat\theta$ be the least square estimate computed by \eqref{eq: least squares}. Let $\calD$ be any distribution with mean $\hat \theta$, and which is supported on a set with diameter bounded by $\frac{1}{256} \norm{P(\theta^\star)}^{-5}$.
    It holds with probability at least $1-\delta$ that 
    \begin{align}
        &C(K_{\mathsf{DR}}(\calD), \theta^\star) - C(K(\theta^\star), \theta^\star) \nonumber \leq \frac{8{\trace\paren{H(\theta^\star}{\mathsf{FI}}(\theta^\star)^{-1})}}{N} + 2\trace\paren{\mathbf{V}(\calD) H(\theta^\star)} \\& + 16\frac{\norm{H(\theta^\star)\mathsf{FI}(\theta^\star)^{-1}}}{N}\log\frac{2}{\delta} + L_{\mathsf{DR}}(\theta^\star)\frac{\norm{\mathsf{FI}(\theta^\star)^{-1}}^{3/2}}{N^{3/2}}, \label{eq:DR upper bound}
    \end{align}
    where $
        L_{\mathsf{DR}}(\theta^\star, \delta) = \mathsf{poly}(d_{\theta}, \max\curly{1, \norm{B(\theta^\star}}, \norm{P(\theta^\star)}, \log\frac{1}{\delta}),$ as long as the number of experiments $N$ satisfies $N \geq \max\curly{N_{\mathsf{ID}}, ~ \frac{c\norm{P(\theta^\star)}^{11} (\norm{B(\theta^\star)} \vee 1)^{16}\paren{d_\theta+\log\frac{2}{\delta}}}{\lambda_{\min}\paren{\mathsf{FI}(\theta^\star)}}}$ for a universal constant $c$.
\end{lemma}

To minimize the above upper bound, choosing a sampling distribution with a variance of zero would be best and would lead to completely canceling the term with $\mathbf{V}(\calD)$. However, this would eliminate any benefits that we hope to see in the low data regime. Instead, to maximize the potential robustness benefits, we should choose the distribution with the most spread, which does not significantly degrade the performance achieved for large $N$ (the regime where the above bound is valid). We therefore propose choosing $\calD$ as a uniform distribution over the confidence ellipsoid constructed for RC \eqref{eq: confidence ellipsoid}. Computing the variance of this quantity demonstrates that the term $\trace\paren{\mathbf{V}(\calD) H(\theta^\star)}$ scales as $\frac{1}{N}\trace\paren{\hat{\mathsf{FI}}^{-1} H(\theta^\star)}\paren{1 + \frac{1}{d_{\theta}}\log\frac{2}{\delta}}$, leading to the following bound. 

\begin{theorem}
    \label{thm: Domain Randomization Upper Bound}
    Under the setting of \Cref{lem: domain randomization general}, let $\calD$ be a uniform distribution over the confidence ellipsoid $G$ defined in \Cref{eq: confidence ellipsoid}. Then it holds with probability at least $1-\delta$ that 
    \begin{align}
        &C(K_{\mathsf{DR}}(\calD), \theta^\star) - C(K(\theta^\star), \theta^\star) \nonumber \leq \frac{40\paren{1+\log\paren{\frac{2}{\delta}}/d_\theta}\trace\paren{H(\theta^\star}{\mathsf{FI}}(\theta^\star)^{-1})}{N} \\
        &+ 16\frac{\norm{H(\theta^\star)\mathsf{FI}(\theta^\star)^{-1}}}{N}\log\frac{2}{\delta} + L_{\mathsf{DR}}(\theta^\star, \delta)\frac{\norm{\mathsf{FI}(\theta^\star)^{-1}}^{3/2}}{N^{3/2}},\label{eq:DR upper bound w/ unif dist}
    \end{align}
    as long as $N$ satisfies the burn-in time of \Cref{lem: domain randomization general}.
\end{theorem}
The burn-in time from \Cref{lem: domain randomization general} provides the powers of $\norm{P(\theta^\star)}$ and $\tau_{B(\theta^\star)}$ listed in the final row of \Cref{tab:sample_efficiency}. The leading term follows from the fact that $L_{\mathsf{DR}}(\theta^\star)$ is multiplied by $N^{-3/2}$, and therefore decays faster. The deviation terms (quantities multiplied by $\log\frac{2}{\delta}$) are not considered leading terms, because they are dominated by the term 
$\frac{1}{N}\trace\paren{H(\theta^\star) \mathsf{FI}(\theta^\star)^{-1}}$ when the bounds are converted from high probability bounds to bounds in expectation.
This leaves a universal constant multiplied by  $\trace\paren{H(\theta^\star) \mathsf{FI}(\theta^\star)^{-1}}$, leading to the characterization of domain randomization with the chosen distribution as sample efficient.

The bound above fails to demonstrate a clear advantage of domain randomization over certainty equivalence in the low-data regime, despite empirical evidence suggesting such benefits (\Cref{fig:result dr lqr}). By designing the sampling distribution to have support on the confidence ellipsoid \eqref{eq: confidence ellipsoid}, we ensure that the true system $\theta^\star$ has positive density in the sampling distribution with high probability. This design raises the hope that domain randomization could reduce the burn-in time. However, we have not been able to prove this property, as we cannot exclude the possibility that for distributions with large support, the domain-randomized controller \eqref{eq: domain randomization} might incur very high costs  near $\theta^\star$ while performing well elsewhere. In contrast, we show in the sequel that robust control can provide such benefits, albeit at the expense of sacrificing asymptotic efficiency.

\subsection{Sample Efficiency of Robust Control}
\label{s: robust control}

We now establish the second row of \Cref{tab:sample_efficiency} by analyzing the efficiency of robust control. Our goal is to demonstrate how robust control addresses the limitations of certainty equivalence with limited data. To achieve this, we introduce a formal definition of robust stabilizability for the system. In this definition, we denote the state covariance of system $\theta$ under controller $K$ by $\Sigma^K(\theta)$. 

\begin{definition}[Robust Stabilizability]
    Let $M\in\R$ and $G$ be a set. $G$ is $M$-robustly stabilizable by CE if $\exists \theta \in G $ such that for all $\theta' \in G$, $\norm{\Sigma^{K(\theta)}(\theta')} \leq M$. 
    Let also $r\in\R$. $\theta$ is $(M,r)$-robustly stabilizable by CE if $\forall A\subseteq\mathcal{B}(\theta, r)$, $A$ is $M$-robustly stablizable by CE. 
\end{definition}

The above definition captures the idea that a CE controller synthesized for a some parameter in a set can stabilize every member of that set. This is stronger than merely assuming the existence of an arbitrary controller that stabilizes all members, and plays a key role in our analysis of robust control. Relaxing this condition is left for future work. Nevertheless, for any stabilizable system, we can choose $r$ sufficiently small to satisfy the condition. Specifically, by Theorem 3 of \citet{simchowitz2020naive}, if $r \leq \frac{1}{256} \norm{P(\theta)}^{-5}$, then $\theta$ is $(M,r)$-robustly stabilizable by CE with $M = 2 \norm{P(\theta)}$. However, the given condition is system-specific and can often be satisfied with larger $r$.
Consider the following example.

\begin{example}
    Consider a scalar linear dynamical system:
        $X_{t+1} = a^\star X_t + b^\star U_t + W_t \quad \forall t\geq 0,$
    where $a^\star = 1.05$ and $b^\star = 1$. Suppose that only $a^\star$ is unknown. Consider the LQR problem defined by $Q=1$, and $R=1000$. Computing $\norm{P(\theta^\star)}$, it holds by \citet{simchowitz2020naive} that the instance is $(M, r)$ stabilizable with $M = 225$ and $r = 2\times 10^{-13}$. However, we can achieve a better characterization for this instance by noting that for any subset $G$ of the interval $[0.3, 1.8]$, synthesizing an LQR controller $k$ using the largest value of the parameter in $G$ ensures that $a + b^\star k < 0.97$ for all $a \in G$, thereby ensuring that $\norm{\Sigma^k(a)} \leq \frac{1}{1-.97^2}\leq 20$ for all $a \in G$. Then the instance is $(20, 0.75)$-robustly stabilizable by CE. 

    To illustrate the impact on control performance, suppose we have an estimate $\hat a = 1.01$. The certainty-equivalent controller derived from this estimate is $k=-0.0424$, which fails to stabilize the true system. In contrast, if we apply robust control over any uncertainty set within the interval 
    $[0.3,1.8]$ that includes $a^\star$, we synthesize a controller that stabilizes the system. Therefore, the robust control procedure selects such a controller to avoid infinite cost.
\end{example}

With the definition of robust stabilizability by certainty equivalence in hand, we proceed to state an upper bound on the excess cost incurred by the robust controller. 

\begin{theorem}
    \label{thm: Robust Control upper bound}
    Suppose the dataset $\curly{(X_t^n, U_t^n, X_{t+1}^n)}_{t=1, n=1}^{T,N}$ is collected from N trajectories of the system \eqref{eq: linear system} via a random control input $U_t \sim \calN(0, \Sigma_u)$. Let $\hat\theta$ be the least square estimate computed by \eqref{eq: least squares}, and $G$ be the confidence ellipsoid of \eqref{eq: confidence ellipsoid}.
    Choose $r > 0$. Let $M$ be the smallest real number such that $\theta^\star$ is $(M, r)$-robustly stabilizable.
    It holds that with probability at least $1-\delta$
    \begin{align}
        &C(K_{RC}(G), \theta^\star)\! -\! C(K(\theta^\star), \theta^\star)\nonumber\!\leq\! \frac{64\paren{d_\theta\!+\!\log\frac{2}{\delta}}\norm{H(\theta^\star)\mathsf{FI}(\theta^\star)^{\!-\!1}}}{N} \!+\! L_{\mathsf{RC}}(\theta^\star\!,\! \delta, M\!)\frac{\norm{\mathsf{FI}(\theta^\star)^{-1}\!}^{3/2}}{N^{3/2}}, 
    \end{align}
    \sloppy where
        $L_{\mathsf{RC}}(\theta^\star, \delta, M) = \mathsf{poly}(d_{\theta}, \max\curly{1, \norm{B(\theta^\star)}}, \norm{P(\theta^\star)}, \log\frac{1}{\delta}, M),$
    as long as the number of experiments satisfies
    $N \geq \max\curly{N_{\mathsf{ID}}, \frac{c\norm{P(\theta^\star)}^4\paren{d_\theta+\log\frac{2}{\delta}}}{\lambda_{\min}\paren{\mathsf{FI}(\theta^\star)}}, \frac{c \paren{d_\theta+\log\frac{2}{\delta}}}{r^2 \lambda_{\min}\paren{\mathsf{FI}(\theta^\star)}}}$ for a universal constant $c$.  
\end{theorem}
A proof of this result is provided in \Cref{s: robust control proof}. Similar to \Cref{thm: Domain Randomization Upper Bound}, the leading term in \Cref{tab:sample_efficiency} is obtained by omitting the deviation term ($\log\frac{2}{\delta}$) and the lower-order term scaling with $N^{-3/2}$
 . While we lack a lower bound for robust control, we conjecture that the leading term is tight. The reasoning behind this conjecture is that robust control selects the worst-case perturbation of the parameter within the confidence set, which naturally leads to dependence on the operator norm, as opposed to the trace observed in alternative approaches.

The burn-in time from in the above theorem demonstrates how the robust stabilizability condition combined with the robust synthesis approach leads to the value reported in \Cref{tab:sample_efficiency}. 


\section{Numerical Experiments}
\label{s: numerical}


\subsection{Linear System}

\begin{wrapfigure}[9]{r}{0.5\textwidth}
    \centering
    \vspace{-25pt} 
    \includegraphics[width=\linewidth]{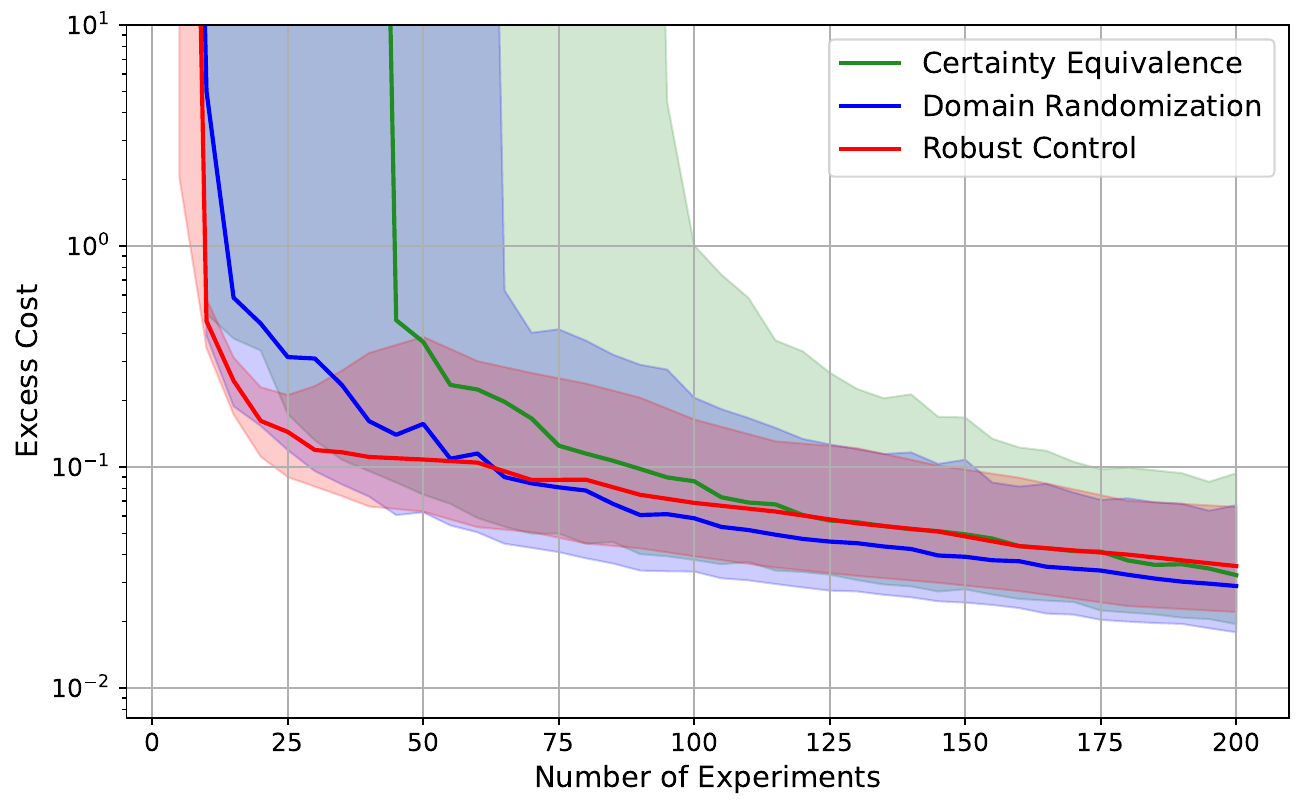}
    \vspace{-21pt} 
    \caption{Excess cost of controllers found via three methods using models fit with various amounts of data.}
    \label{fig:result dr lqr}
\end{wrapfigure}
\renewcommand{\arraystretch}{1.0}
We validate the trends predicted in \Cref{tab:sample_efficiency} through a case study on the linear system 
\vspace{-3pt}
\begin{align}
    &A = \bmat{
        1.01 & 0.01 & 0 \\
        0.01 & 1.01 & 0.01 \\
        0 & 0.01 & 1.01}, \nonumber\\
    &B = I, Q = 10^{-3}I, R = I. \label{eq:numerical experiments params}
\end{align}
We first estimate $A$ and $B$ using least squares identification, then synthesize CE, DR, and RC controllers based on the identified models. Further details are provided in \Cref{s: implementation details}.

In \Cref{fig:result dr lqr}, we plot the median and shade 25 \% to 75 \% quantile over 500 random seeds. 
This result reveals two key observations. First, both DR and RC stabilize the system with fewer experiments than CE. While the improved performance of DR in the low-data regime lacks concrete theoretical justification, \Cref{thm: Robust Control upper bound} supports this trend for RC. Second, after an initial period, DR converges faster than RC, eventually matching the convergence rate of CE. This aligns with the conclusion of \Cref{thm: Domain Randomization Upper Bound} and is consistent with the conceptual illustration in \Cref{fig:conceptual figure}.

\subsection{Pendulum}
\begin{wrapfigure}{r}{0.5\textwidth}
    \centering
    \vspace{-90pt}
    \includegraphics[width=\linewidth]{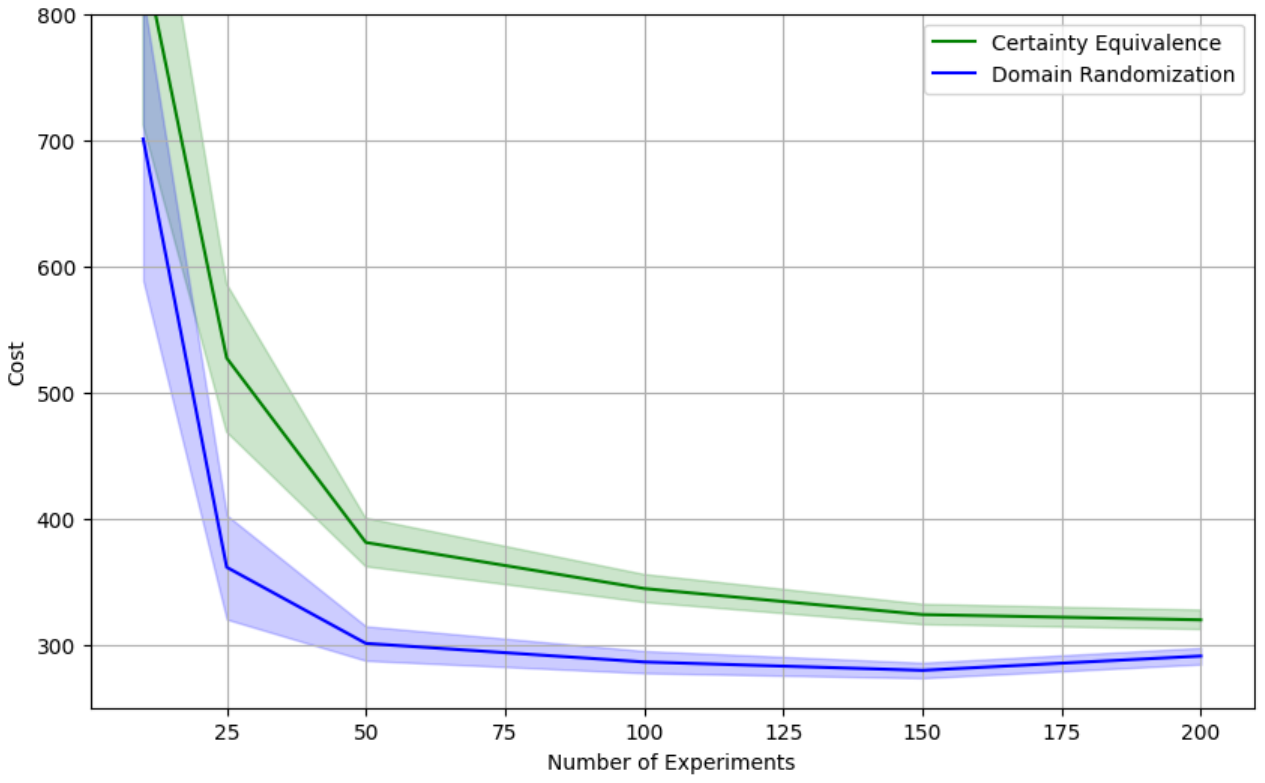}
    \vspace{-90pt}
    \caption{Cost of CE and DR controllers based on models fit with various amounts of data.}
    \label{fig:pendulum}    
\end{wrapfigure}

We extend the approach to nonlinear systems of pendulum by applying receding horizon control. For CE, planning occurs by minimizing a finite horizon cost for a trajectory planned using the nominal system estimate. For domain randomization, planning occurs by minimizing the average finite horizon cost for trajectories planned with a collection of systems sampled from a distribution around the system estimate. 

The result is shown in \Cref{fig:pendulum}, which plots the mean and standard error over 100 random seeds. As observed in the linear system, the trends that (i) DR outperforms CE in the low-data regime and (ii) the convergence rate of CE and DR matches hold even for the nonlinear system. 
Implementation details are deferred to the \Cref{s: implementation details}.

\section{Discussion}

\paragraph{Algorithmic considerations:}
Our numerical experiments use a gradient-based algorithm for DR, detailed in \Cref{alg: dr lqr}. This algorithm builds on the scenario approach of \citet{vidyasagar2001randomized} and the policy gradient method introduced for the linear quadratic regulator by \citet{fazel2018global}. The algorithm initially samples a number of scenarios from the distribution. At each iteration, the gradient update is performed on the cost summed over all scenarios. Since the control cost gradient becomes infinite if the current iterate does not stabilize a system, we incorporate only the gradients for system which the current iterate stabilizes. While we lack a formal convergence guarantee for this algorithm, numerical experiments indicate that it converges with a sufficiently small step size.
\begin{algorithm}[t]
 \caption{Domain Randomized Policy-Gradient for the Linear Quadratic Regulator} 
 \label{alg: dr lqr}
\begin{algorithmic}[1]
\vspace{-2pt}
\State \textbf{Input: } Randomization distribution $\calD$, estimate $\hat \theta$, stepsize $\eta$, $\#$ iterations $M$, $\#$ scenarios $N$
\State \textbf{Initialize: } $\hat K_1 \gets K_{CE}(\hat \theta)$,
\State \textbf{Sample: }
 $N$ scenarios $\theta_1, \dots, \theta_N \sim \calD$
    \For{$i = 1, 2, \dots, M$}  $\backslash\backslash$ Gradient descent on stable scenarios 
    \State $\hat K_{i+1} = \hat K_i - \eta \sum_{j=1}^N \nabla C(\hat K_i, \theta_{j}) \mathbf{1}(\rho(A(\theta_{j}) + B(\theta_{j}) \hat K_i) <1)$
    \EndFor
\State \textbf{Return: } $\hat K_{M}$. 
\end{algorithmic}
\vspace{-2pt}
\end{algorithm}

This algorithm raises numerous questions which may be fruitful directions for future work. 
\begin{itemize} [noitemsep,nolistsep,leftmargin=*]
    \item \textbf{Convergence analysis of \Cref{alg: dr lqr}}: Extending the convergence analysis of policy gradient methods for LQR by \citet{fazel2018global, hu2023toward} to the proposed algorithm could provide theoretical guarantees for convergence to the solution of \eqref{eq: domain randomization}. Such an analysis would offer strong evidence of the algorithm's applicability beyond the toy numerical example presented here.
    \item \textbf{Alternative choices of distribution:} We proposed sampling from a uniform distribution over the confidence ellipsoid \eqref{eq: confidence ellipsoid} to maximize the spread of the sampling distribution without sacrificing asymptotic efficiency (\Cref{thm: Domain Randomization Upper Bound}). Exploring alternative distributions, such as truncated normal distributions, could reveal differences in empirical performance. These alternatives might also enable refined analyses of \Cref{thm: Domain Randomization Upper Bound}, particularly with respect to burn-in time.
    \item  \textbf{Extension to nonlinear systems:}  \Cref{alg: dr lqr} can, in principle, be extended to nonlinear systems, provided that gradients of the control objective can be obtained via Monte Carlo sampling. The least-squares analysis for CE in \Cref{tab:sample_efficiency} also extends nonlinear systems \citep{lee2024active}, suggesting that the proposed domain randomization approach could also be effective for such systems. This may yield sample efficiency guarantees analogous to those studied here.
\end{itemize}

\paragraph{Theoretical extensions: } There are additionally numerous directions to tighten the analysis and generalize the setting. 
\begin{itemize}[noitemsep,nolistsep,leftmargin=*]
    \item \textbf{Improved burn-in time for domain randomization: }
    While this work empirically demonstrates that DR can stabilize the system even in the low-data regime, the burn-in time derived in \Cref{lem: domain randomization general} does not reflect this advantage, as it is larger than  that of CE in \Cref{thm: certainty equivalence bound}. The only known way to improve the burn-in time is by adopting a robustly stabilizable condition, which considers the worst case but results in a conservative upper bound on the excess cost, as shown in \Cref{thm: Robust Control upper bound}. A promising direction for future work is to tighten the analysis for burn-in requirements and higher order terms of DR while maintaining its asymptotic efficiency.
    \item \textbf{Robust control lower bound:} Our upper bounds feature a gap between the asymptotic convergence rate of RC and that of CE and DR. We conjecture that this gap is fundamental; however, a lower bound on the sample efficiency of RC would be required to formalize this.
    \item \textbf{Misspecification:}  This work has focused on the use of DR to address model uncertainty arising from variance in model fitting. Specifically, the assumption of the dynamics in \eqref{eq: linear system} imposes a realizability condition, ensuring that a suitably constructed distribution $\calD$ contains $\theta^\star$ in its support with high probability. However, a key explanation for the empirical success of DR in many robotics applications is its robustness to model misspecification \citep{tobin2017dr}. Investigating this theoretically represents a promising direction for future work.
\end{itemize}

\section{Conclusion}

By analyzing the sample efficiency of learning the linear quadratic regulator via domain randomization and robust control, our work provides insights into the tradeoffs present for approaches to incorporate uncertainty quantification into learning-enabled control. Our analysis demonstrates that if one is strategic about the design of the sampling distribution, then the benefits of domain randomization over robust control may extend beyond computational considerations, and to the sample efficiency. This is particularly exciting due to the prominence of domain randomization in practice for robot learning. We believe that this line of analysis exposes a wide spread of interesting questions regarding the use of domain randomization for learning-enabled control. 

\acks{We thank Manfred Morari, Anastasios Tsiamis, Ingvar Ziemann, and Thomas Zhang for several instructive conversations. 
TF is supported by JASSO Exchange Support program and UTokyo-TOYOTA Study Abroad Scholarship. TF and GP are supported in part by NSF Award SLES 2331880 and NSF TRIPODS EnCORE 2217033.  BL and NM are supported by NSF Award SLES-2331880, NSF CAREER award ECCS-2045834 and AFOSR Award FA9550-24-1-0102.}

\bibliography{refs}
\clearpage
\appendix

\section{Notation}
The Euclidean norm of a vector $x$ is denoted $\norm{x}$. For a matrix $A$, the spectral norm is denoted $\norm{A}$, and the Frobenius norm is denoted $\norm{A}_F$. 
A symmetric, positive semi-definite matrix $A = A^\top$ is denoted $A \succeq 0$.  $A \succeq B$ denotes that $A-B$ is positive semi-definite. Similarly, a symmetric, positive definite matrix $A$ is denoted $A \succ 0$. 
The minimum eigenvalue of a symmetric, positive semi-definite matrix $A$ is denoted $\lambda_{\min}(A)$. For a positive definite matrix $A$, we define the $A$-norm as $\norm{x}_A^2 = x^\top A x$. 
The gradient of a scalar valued function $f: \R^n \to \R$ is denoted $\nabla f$, and the Hessian is denoted $\nabla^2 f$. 
The Jacobian of a vector-valued function $g: \R^n \to \R^m$ is denoted $D g$, and follows the convention for any $x\in\R^n$, the rows of $D g(x)$ are the transposed gradients of $g_i(x)$.
The $p^{th}$ order derivative of $g$  is  denoted by $D^{p} g$. Note that for $p \geq 2$, $D^{p} g(x)$ is a tensor for any $x\in\R^{n}$. 
The operator norm of such a tensor is denoted by $\norm{D^p g(x)}_{\mathsf{op}}$. 
For a function $f: \mathsf{X} \to \R^{\dy}$, we define $\norm{f}_{\infty} \triangleq \sup_{x \in \mathsf{X}} \norm{f(x)}$. 
A Euclidean norm ball of radius $r$ centered at $x$ is denoted $\calB(x,r)$. The state covariance matrix of the system under controller K is denoted as $\Sigma^K(\theta) \triangleq \dlyap((A(\theta)+B(\theta)K)^T, I )$. 
The solution to the discrete algebraic Ricatti equation satisfies $P \succeq I$ as long as $Q\succeq I$. 
The Hessian of the objective function is calculated as 
\begin{align*}
 H(\theta) &\triangleq \nabla^2_\theta C(K(\theta), \theta) = \mathsf{D}_{\theta}\VEC K(\theta)^T(\Psi(\theta)\kron\Sigma_X^{K(\theta)}(\theta))\mathsf{D}_{\theta}\VEC K(\theta)
\end{align*}
where $\Psi(\theta) \triangleq B(\theta)^TP(\theta)B(\theta) + R$.

\section{Perturbation Analysis}


Here we present a number of perturbation results that we re-use throughout our analysis.


\begin{lemma}[Performance Difference Lemma, Lemma 12 of \citet{fazel2018global}]
    \label{lem: performance difference}
    Let $\theta$ denote the parameter for a dynamical system, and $K$ be an arbitrary gain that stabilizes this system. Then it holds that
    \begin{align*}
        C(K, \theta) - C(K(\theta), \theta) = \trace\paren{(K-K(\theta) \Sigma^K(\theta) (K-K(\theta))^\top \Psi(\theta)}, 
    \end{align*}
    where we recall that $\Sigma^K(\theta)$ is the state covariance of the system under controller $K$. 
\end{lemma}

\begin{lemma}[Lyapunov Perturbation]
    \label{lem: lyap perturbation}
    Let $A_1, A_2 \in \R^{d \times d}$ satisfy  $\rho(A_1) < 1$ and $\rho(A_2) < 1$. Let $Q$ be a $d$ dimensional positive definite matrix. Define $P_1 = \dlyap(A_1, Q)$, and $P_2 = \dlyap(A_2, Q)$. Then it holds that 
    \begin{align*}
        \norm{P_1 - P_2} \leq \frac{1}{\lambda_{\min}(Q)}\norm{P_1} \norm{P_2}\paren{2\norm{A_2} \norm{A_1-A_2}+\norm{A_1 - A_2}^2}. 
    \end{align*}
\end{lemma}
\begin{proof}
    By definition of $P_1$ and $P_2$, it holds that 
    \begin{align*}
        P_1 - P_2 &= A_1^\top P_1 A_1 - A_2^\top P_2 A_2 \\
        &= (A_2 + (A_1-A_2))^\top P_1 (A_2 + (A_1-A_2)) - A_2^\top P_2 A_2 \\
        &= \dlyap(A_2, A_2^\top P_1 (A_1-A_2) +(A_1-A_2)^\top P_1 A_2 +  (A_1-A_2)^\top P_1(A_1-A_2)) \\
        &\preceq \dlyap(A_2, I) \norm{A_2^\top P_1 (A_1-A_2) +(A_1-A_2)^\top P_1 A_2 +  (A_1-A_2)^\top P_1(A_1-A_2)}.
    \end{align*}
    Note that $\dlyap(A_2, I)\preceq \frac{1}{\lambda_{\min}(Q)} P_2$. The result then follows by the triangle inequality and submultiplicativity.
\end{proof}

\begin{lemma}
    \label{lem: cov lower bound}
    Fix a system $\theta$ and two stabilizing controllers $K_1$ and $K_2$. As long as $\norm{B(\theta)(K_1 - K_2)} \leq \frac{1}{12 \norm{\Sigma^{K_1}(\theta)}^{5/2}}$, it holds that 
    \begin{align*}
        \Sigma^{K_2}(\theta) \succeq \frac{1}{2} \Sigma^{K_1}(\theta).
    \end{align*}
\end{lemma}
\begin{proof}
    The result follows by applying the reverse triangle inequality along with \Cref{lem: lyap perturbation}.
\end{proof}

\begin{lemma}[Riccati Perturbation, Proposition 4 and 6 of \citet{simchowitz2020naive}]
    \label{lem: Riccati perturbation}
    Let $\theta_1$ denote the parameter for a stabilizable system, and $\theta_2$ denote the parameter for an alternate system. Suppose that $\norm{\theta_1-\theta_2} \leq \frac{1}{16}\norm{P(\theta_1)}^{-2}$. Then the system described by $\theta_2$ is stabilizable, and the following inequalities hold.
    \begin{itemize}
        \item $\norm{P(\theta_2)}\leq \sqrt{2} \norm{P(\theta_1)}$
        \item $\max\curly{\norm{K(\theta_2) - K(\theta_1)}, \norm{B(\theta_1) (K(\theta_2) - K(\theta_1))}
        }\leq  32\norm{P(\theta_1)}^{7/2} \norm{\theta_1-\theta_2}.$
        \item $\norm{P(\theta_2) - P(\theta_1)} \leq 8\sqrt{2}\norm{P(\theta_1)}^3\norm{\theta_1-\theta_2}$
    \end{itemize}
\end{lemma}

\begin{lemma}[Simplifying inequalities]
    \label{lem: simplifying inequalities}
    Let $\theta$ be a parameter describing a stabilizable system instance. 
    Define $\tau_{B(\theta)} = \max\curly{1, \norm{B(\theta)}}$.
    The following inequalities hold 
    \begin{itemize}
        \item $\norm{\Sigma^{K(\theta)}(\theta)} \leq \norm{P(\theta)}$
        \item $\norm{A(\theta) + B(\theta) K(\theta)} \leq \norm{P(\theta)}^{1/2}$. 
        \item $\norm{K(\theta)}\leq\norm{P(\theta)}^{1/2}$
        \item $\norm{\Psi(\theta)} \leq 2 \tau_{B(\theta)}^2 \norm{P(\theta)}$.
        \item $\norm{\Psi(\theta)\kron\Sigma_X^{K(\theta)}(\theta)} \leq 2\tau^2_{B(\theta)}\norm{P(\theta)}^2$. 
    \end{itemize}
\end{lemma}
\begin{proof}
    The first inequality follows by observing that 
    \begin{align*}
        \norm{\Sigma^{K(\theta)}(\theta)} &= \norm{\dlyap((A+BK(\theta))^\top, I)} 
        \leq \norm{\dlyap((A+BK(\theta)), Q + K(\theta)^\top R K(\theta))},
    \end{align*}
    by the fact that $Q \succeq I$. The second inequality follows by noting that 
    \begin{align*}
        \norm{(A(\theta) + B(\theta) K(\theta))^\top (A(\theta) + B(\theta) K(\theta))}^{1/2} \leq \norm{\dlyap(A(\theta)+B(\theta)K(\theta), I)}^{1/2}.    
    \end{align*}
    The third inequality follows from
    \begin{align*}
        \norm{K(\theta)} \leq \norm{Q+K(\theta)^TRK(\theta)}^{1/2} \leq \norm{\dlyap(A(\theta)+B(\theta)K(\theta), Q+K(\theta)^TRK(\theta))}^{1/2}.
    \end{align*}
    The fourth and fifth inequality follow from
    \begin{align*}
        \norm{\Psi(\theta)\kron\Sigma_X^{K(\theta)}(\theta)} &\leq \norm{\Psi(\theta)}\norm{\Sigma_X^{K(\theta)}(\theta)} \leq (\norm{B(\theta)}^2+1)\|P(\theta)\|^2\leq 2\tau^2_{B(\theta)}\norm{P(\theta)}^2. 
    \end{align*}
    where we used $\norm{X\kron Y}\leq\norm{X}\norm{Y}$.
\end{proof}

\begin{lemma}[Certainty Equivalent Stabilization]
    \label{lem: CE stabilization}
    Let $\theta_1$ denote the parameter for a stabilizable system, and $\theta_2$ denote the parameter for an alternate system. Suppose that $\norm{\theta_1-\theta_2} \leq \frac{1}{256} \norm{P(\theta_1)}^{-5}$. Then the system described by $\theta_2$ is stabilizable
    \begin{align*}
        \norm{\Sigma^{K(\theta_2)}(\theta_1)} \leq 2 \norm{P(\theta_1)}. 
    \end{align*}
\end{lemma}
\begin{proof}
    To verify this fact, first apply
 \Cref{lem: lyap perturbation} and \Cref{lem: simplifying inequalities} to find
 \begin{align*}
    & \norm{\Sigma^{K(\theta_2)}(\theta_1)} \leq  \norm{\Sigma^{K(\theta_1)}(\theta_1)} + \norm{\Sigma^{K(\theta_1)}(\theta_1)}\norm{\Sigma^{K(\theta_2)}(\theta_1)} \\
    & \paren{2\norm{A(\theta_1) + B(\theta_1) K(\theta_1)} \norm{B(\theta_1) (K(\theta_2)-K(\theta_1)} + \norm{B(\theta_1) (K(\theta_2)-K(\theta_1)}^2  } \\
    &\leq \norm{P(\theta_1)} + \norm{P(\theta_1)} \norm{\Sigma^{K(\theta_2)}(\theta_1)} \paren{2 \norm{P(\theta_1)}^{1/2} \norm{B(\theta_1)(K(\theta_2)  - K(\theta_1))} +\norm{B(\theta_1)(K(\theta_2)  - K(\theta_1))}^2}.
 \end{align*}
 By \Cref{lem: Riccati perturbation} it holds that $\norm{B(\theta_1) (K(\theta_2) - K(\theta_1))} \leq  32 \norm{P(\theta_1)}^{7/2} \norm{\theta_2 - \theta_1}$. Leveraging that $\norm{\theta_2 - \theta_1}\leq\frac{1}{256} \norm{P(\theta_1)}^{-5}$, we conclude 
 \begin{align*}
     \norm{\Sigma^{K(\theta_2)}(\theta_1)} \leq \norm{P(\theta_1)} + \frac{1}{2}\norm{\Sigma^{K(\theta_2)}(\theta_1)}.
 \end{align*}
 Rearranging provides the desired inequality. 
 \end{proof}

\begin{lemma}[Bound on the first and second derivative of K]
    \label{lem: Bound on K' and K''}
    Let $\theta_1, \theta_2$ be the parameters describing two stabilizable systems satisfying $\norm{\theta_1-\theta_2}\leq \frac{1}{16} \norm{P(\theta_1)}^{-2}$. Then for any $t\in[0,1]$, the first and second derivative is bounded as 
    \begin{align*}
        \|\mathsf{D}\VEC K(\tilde\theta)\|_{\mathsf{op}}\leq 24\|P\paren{\theta_1}\|^{7/2}, ~ 
        \|\mathsf{D}^2\VEC K(\tilde\theta)\|_{\mathsf{op}}\leq 4000\|P\paren{\theta_1}\|^{15/2},
    \end{align*}
    where $\tilde \theta = t\theta_1 + (1-t)\theta_2$.
\end{lemma}
\begin{proof}
    It suffices to consider the quantity $K(t)$ defined as
    \begin{align*}
        &K(t) \triangleq K(t\theta_1 + (1-t)\theta_2). 
    \end{align*}
    Additionally define
    \begin{align*}
        P(t) \triangleq P(t\theta_1 + (1-t)\theta_2).
    \end{align*}
    By the proof of Lemma 3.2, C.5 of \citet{simchowitz2020naive} and keeping track of the degree of the term $\norm{P(\theta_1)}$, we get
    \begin{align*}
        \|P'(t)\|&\leq 4\|P(t)\|^3 , 
        ~ \|P''(t)\|\leq178\|P(t)\|^7 \\
        \|K'(t)\|&\leq7\|P(t)\|^{7/2} , ~\|K''(t)\|\leq290\|P(t)\|^{15/2}.
    \end{align*}
    From \Cref{lem: Riccati perturbation}, it holds that $\norm{P(t)} \leq \sqrt{2}\norm{P(\theta_1)}$. Therefore, by noting that the derivatives of $K(t)$ are directional derivatives of $K(\theta)$, it follows that
    \begin{align*}
        &\|\mathsf{D}\VEC K(\tilde\theta)\|_{op}\leq \sup_{t\in[
        0,1]}7\|P(t)\|^{7/2} < 24\|P\paren{\theta_1}\|^{7/2} \\
        &\|\mathsf{D}^2\VEC K(\tilde\theta)\|_{op}\leq \sup_{t\in[
        0,1]}290\|P(t)\|^{15/2} <
        4000\|P\paren{\theta_1}\|^{15/2}.
    \end{align*}
\end{proof}

Leveraging the above result along with a Taylor expansion leads to the following lemma. 
\begin{lemma} [LQR Taylor Expansion]
    \label{lem: LQR Taylor expansion}
    Let $\theta_1, \theta_2$ be the parameters describing two stabilizable systems satisfying $\norm{\theta_1-\theta_2}\leq \frac{1}{16} \norm{P(\theta_1)}^{-2}$. It holds that
    \begin{align*}
        \VEC (K(\theta_2) - K(\theta_1)) &= \mathsf{D}_{\theta}\VEC K(\theta_1)[\theta_2 - \theta_1] + R,
    \end{align*}
    where $R$ is the remainder term that satisfies
    \begin{align*}
        \norm{R} &\leq \frac{1}{2}\sup_{\tilde\theta\in[\theta_1, \theta_2]}\|\mathsf{D}^2\VEC K(\tilde\theta)\|_{op}\|\theta_2-\theta_1\|^2 \leq 2000\|P_{\theta_1}\|^{15/2}\|\theta_2-\theta_1\|^2.
    \end{align*}
\end{lemma}

\begin{lemma}[Suboptimality Gap Bound]
    \label{lem: excess cost decomposition}
    Let $K$ be any controller that stabilizes $\theta$. Then it holds that
    \begin{align*}
        &C(K,\theta) - C(K(\theta), \theta) \\
        &\leq \trace ((K-K(\theta))\Sigma_X^{K(\theta)}(\theta)(K-K(\theta))^T\Psi(\theta))\\
        &+ 2 \norm{P(\theta)}^2 \|\Sigma_X^{K}(\theta)\|\tau_{B(\theta)}^3\|(K-K(\theta))\|^3(\|B(\theta)(K-K(\theta))\| + 2 \norm{P(\theta)}^{1/2}).
    \end{align*}
\end{lemma}
\begin{proof}
    By \Cref{lem: performance difference}, 
    \begin{align*}
        &C(K,\theta) - C(K(\theta), \theta) \\
        &\leq \trace ((K-K(\theta))\Sigma_X^{K(\theta)}(\theta)(K-K(\theta))^T\Psi(\theta) + (K-K(\theta))(\Sigma_X^K(\theta) - \Sigma_X^{K(\theta)}(\theta))(K-K(\theta))^T\Psi(\theta))
\end{align*}
where $\Sigma_X^K(\theta) = \dlyap((A(\theta)+B(\theta) K)^T, I)$. 
By applying \Cref{lem: lyap perturbation} to the second term, we get
\begin{align*}
    &\trace(K-K(\theta))(\Sigma_X^K(\theta) - \Sigma_X^{K(\theta)}(\theta))(K-K(\theta))^T\Psi(\theta))\\
    &\leq \|\Sigma_X^K(\theta)\|\|\Sigma_X^{K(\theta)}(\theta)\|\|K-K(\theta)\|^3\|\Psi(\theta)\|\|B(\theta)\|(\|B(\theta)(K-K(\theta))\| + 2\|A(\theta)+B(\theta) K(\theta)\|).
\end{align*}
To conclude, we apply the inequalities of \Cref{lem: simplifying inequalities} to $\norm{A(\theta) + B(\theta) K(\theta)}$, $\norm{\Sigma^{K(\theta)}(\theta)}$, and $\norm{\Psi(\theta)}$. 
\end{proof}

\begin{lemma}[Taylor Expansion Substitution for Suboptimality Gap]
    \label{lem: cost gap taylor substitution}
    Let $K(\theta_1), K(\theta_2)$ be any certainty equivalence controller that stabilize $\theta_1, \theta_2$, respectively, where $\theta_1, \theta_2$ satisfy $\norm{\theta_1-\theta_2}\leq \frac{1}{16} \norm{P(\theta_1)}^{-2}$. 
    Then it holds that
    \begin{align*}
        &\trace ((K(\theta_2)-K(\theta_1))\Sigma_X^{K(\theta_1)}(\theta_1)(K(\theta_2)-K(\theta_1))^T\Psi(\theta_1)) \\
        &\leq \|\theta_2-\theta_1\|_{H(\theta_1)}^2 + 2e5\tau_{B(\theta_1)}^2\|P(\theta_1)\|^{13}\|\theta_2-\theta_1\|^3 + 8e6\tau_{B(\theta_1)}^2\|P(\theta_1)\|^{17}\|\|\theta_2-\theta_1\|^4.
    \end{align*}
\end{lemma}
\begin{proof}
    By \Cref{lem: simplifying inequalities}, \Cref{lem: Bound on K' and K''} and \Cref{lem: LQR Taylor expansion}, it follows that
    \begin{align*}
        &\trace ((K(\theta_2)-K(\theta_1))\Sigma_X^{K(\theta_1)}(\theta_1)(K(\theta_2)-K(\theta_1))^T\Psi(\theta_1)) \\
        &= \VEC (K(\theta_2) - K(\theta_1))^T(\Psi(\theta_1)\kron\Sigma_X^{K(\theta_1)}(\theta_1))\VEC(K(\theta_2)-K(\theta_1)) \\
        &\leq [\theta_2 - \theta_1]^T\mathsf{D}_{\theta}\VEC K(\theta_1)^T(\Psi(\theta_1)\kron\Sigma_X^{K(\theta_1)}(\theta_1))\mathsf{D}_{\theta}\VEC K(\theta_1)[\theta_2 -\theta_1] \\
        &\hspace{5mm} + \sym([\theta_2 - \theta_1]^T\mathsf{D}_{\theta}\VEC K(\theta_1)^T(\Psi(\theta_1)\kron\Sigma_X^{K(\theta_1)}(\theta_1))R) + R^T(\Psi(\theta_1)\kron\Sigma_X^{K(\theta_1)}(\theta_1))R \\
        &\leq \|\theta_2-\theta_1\|_{H(\theta_1)}^2 + 2e5\tau_{B(\theta_1)}^2\|P(\theta_1)\|^{13}\|\theta_2-\theta_1\|^3 + 8e6\tau_{B(\theta_1)}^2\|P(\theta_1)\|^{17}\|\|\theta_2-\theta_1\|^4.
    \end{align*} 
    where it follows from $\trace(A^T, B) = \VEC (A)^T\VEC B$ and $\VEC (ABC) = (C^T\kron A)\VEC (B)$ in the first equality, and we let the operator $\sym$ denote $\sym(A) = A+A^T$. 
\end{proof}

\begin{lemma}[Perturbation on $B(\theta), H(\theta)$]
    \label{lem: helper lemma for RC}
    Let $\theta_1, \theta_2$ be the parameters describing two stabilizable systems satisfying $\norm{\theta_1 - \theta_2}\leq\frac{1}{16}\norm{P(\theta_1)}^{-2}$. Then it holds that
    \begin{itemize}
        \item $\norm{B(\theta_2)}\leq \norm{B(\theta_1)} + \norm{\theta_1 - \theta_2}$
        \item $\norm{\Psi(\theta_2) - \Psi(\theta_1)} \leq 15\tau^2_{B(\theta_1)}\norm{P(\theta_1)}^3\norm{\theta_1-\theta_2}$
        \item $\norm{H(\theta_2)-H(\theta_1)}\leq 5e6\tau_{B(\theta_1)}^2\norm{P(\theta_1)}^{17}\norm{\theta_1-\theta_2}$.
        \item $\norm{H(\theta_2)} \leq 8e3\tau^2_{B(\theta_1)}\norm{P(\theta_1)}^9$
    \end{itemize}
\end{lemma}
\begin{proof}
    By the triangle inequality,
    \begin{align*}
        \norm{B(\theta_2)} = \norm{B(\theta_1) + (B(\theta_2) - B(\theta_1))} \leq \norm{B(\theta_1)} + \norm{\theta_1 - \theta_2}.
    \end{align*}
    It follows that
    \begin{align*}
        \tau^2_{B(\theta_2)} &\leq \max\{\norm{B_2}^2, 1\} \leq \max\{\norm{B(\theta_1)}^2+2\norm{B(\theta_1)}\norm{\theta_1-\theta_2} + \norm{\theta_1-\theta_2}, 1\} \leq 2\tau_{B(\theta_1)}^2 \\
        \tau^3_{B(\theta_2)} &\leq \max\{\norm{B_2}^3, 1\} \leq \max\{\norm{B(\theta_1)}^3+3\norm{B(\theta_1)}^2\norm{\theta_1-\theta_2} + 3\norm{B(\theta_1)}\norm{\theta_1-\theta_2}^2
        \norm{\theta_1-\theta_2}^3, 1\} \\
        &\leq 2\tau_{B(\theta_1)}^3
    \end{align*}
    where we applied $\norm{\theta_1-\theta_2} \leq \frac{1}{16}\norm{P(\theta_1)}^{-2}$ in the last inequality.
    For the second inequality, it holds from \Cref{lem: Riccati perturbation} that
    \begin{align*}
        &\norm{\Psi(\theta_2) - \Psi(\theta_1)} \\
        &\leq \norm{\{B(\theta_1) + B(\theta_2)-B(\theta_1)\}^TP(\theta_2)\{B(\theta_1) + B(\theta_2)-B(\theta_1)\} - B(\theta_1)^TP(\theta_1)B(\theta_1)} \\
        &\leq \norm{B(\theta_1)^T\paren{P(\theta_2)-P(\theta_1)}B(\theta_1)} + \norm{\sym\paren{B(\theta_1)^TP(\theta_2)\paren{B(\theta_2)-B(\theta_1)}}} \\
        &\quad+ \norm{\paren{B(\theta_2)-B(\theta_1)}^TP(\theta_2)\paren{B(\theta_2)-B(\theta_1)}} \\
        &\leq 8\sqrt{2}\norm{B(\theta_1)}^2\norm{P(\theta_1)}^3\norm{\theta_1-\theta_2} + 
        2\sqrt{2}\norm{B(\theta_1)}\norm{B(\theta_2)-B(\theta_1)}\norm{P(\theta_1)}  \\
        &\quad + \sqrt{2}\norm{B(\theta_2)-B(\theta_1)}^2\norm{P(\theta_1)} \\
        &\leq 15\tau^2_{B(\theta_1)}\norm{P(\theta_1)}^3\norm{\theta_1-\theta_2}.
    \end{align*}
    where we applied $\norm{\theta_1-\theta_2} \leq \frac{1}{16}\norm{P(\theta_1)}^{-2}$ in the last inequality. 
    Next, let $K'(\theta) \triangleq \mathsf{D}_{\theta}\VEC K(\theta)$. Then from \Cref{lem: Bound on K' and K''}, 
    \begin{align*}
        \norm{K'(\theta_2)-K'(\theta_1)} 
        &= \sup_{\tilde\theta\in[\theta_1, \theta_2]}\norm{K''(\tilde\theta)}\norm{\theta_1-\theta_2}
        \leq 4000\norm{P(\theta_1)}^{15/2}\norm{\theta_1-\theta_2}.
    \end{align*}
    From \Cref{lem: Riccati perturbation} and \Cref{lem: simplifying inequalities}
    \begin{align*}
        &\norm{A(\theta_2)+B(\theta_2)K(\theta_2) - A(\theta_1)+B(\theta_1)K(\theta_1)} \\
        &\leq \norm{A(\theta_2)-A(\theta_1)} + \norm{(B(\theta_2)-B(\theta_1))K(\theta_2)} + \norm{B(\theta_1)(K(\theta_2)-K(\theta_1))} \\
        &\leq \norm{\theta_1-\theta_2} + \norm{\theta_1-\theta_2}\norm{P(\theta_2)}^{1/2} + 32\norm{P(\theta_1)}^{7/2}\norm{\theta_1-\theta_2} \\
        &\leq 35\norm{P(\theta_1)}^{7/2}\norm{\theta_1-\theta_2}. 
    \end{align*}
    For $\Sigma_X^{K(\theta)}(\theta)$, from \Cref{lem: lyap perturbation} and above calculation, 
    \begin{align*}
        &\norm{\Sigma_X^{K(\theta_2)}(\theta_2) - \Sigma_X^{K(\theta_1)}(\theta_1)} \\
        &\leq \norm{\Sigma_X^{K(\theta_1)}(\theta_1)}\norm{\Sigma_X^{K(\theta_2)}(\theta_2)}\paren{70\norm{P(\theta_1)}^{4}\norm{\theta_1-\theta_2} + 35^2\norm{P(\theta_1)}^{7}\norm{\theta_1-\theta_2}^2} \\
        &\leq \sqrt{2}\norm{P(\theta_1)}^2\paren{70\norm{P(\theta_1)}^{4}\norm{\theta_1-\theta_2} + 35^2\norm{P(\theta_1)}^{7}\norm{\theta_1-\theta_2}^2} \\
        &\leq 225\norm{P(\theta_1)}^{7}\norm{\theta_1-\theta_2},
    \end{align*}
    where we applied $\norm{\theta_1-\theta_2} \leq \frac{1}{16}\norm{P(\theta_1)}^{-2}$ in the last inequality.
    Now let $M(\theta)$ denote $M(\theta) \triangleq \Psi(\theta)\kron\Sigma_X^{K(\theta)}(\theta)$. 
    Then we get
    \begin{align*}
        &\norm{M(\theta_2) - M(\theta_1)} \\
        &= \norm{\paren{\Psi(\theta_1) + \Psi(\theta_2) - \Psi(\theta_1)}\kron\paren{\Sigma_X^{K(\theta_1)}(\theta_1) + \Sigma_X^{K(\theta_2)}(\theta_2) - \Sigma_X^{K(\theta_1)}(\theta_1)} - \Psi(\theta_1)\kron\Sigma_X^{K(\theta_1)}(\theta_1)} \\
        &\leq \norm{\Psi(\theta_2) - \Psi(\theta_1)}\norm{\Sigma_X^{K(\theta_1)}(\theta_1)} + \norm{\Psi(\theta_1)}\norm{\Sigma_X^{K(\theta_2)}(\theta_2) - \Sigma_X^{K(\theta_1})(\theta_1)} \\
        &\hspace{5mm}+ \norm{\Psi(\theta_2) - \Psi(\theta_1)}\norm{\Sigma_X^{K(\theta_2)}(\theta_2) - \Sigma_X^{K(\theta_1)}(\theta_1)} \\
        &\leq 4e3\tau_{B(\theta_1)}^2\norm{P(\theta_1)}^{10}\norm{\theta_1-\theta_2}
    \end{align*}
    where we used $\norm{X\kron Y}\leq\norm{X}\norm{Y}$.
    Thus from \Cref{lem: Bound on K' and K''}
    \begin{align*}
        &\norm{H(\theta_2) - H(\theta_1)} \\
        &= \norm{K'(\theta_2)^TM(\theta_2)K'(\theta_2) - K'(\theta_1)^TM(\theta_1)K'(\theta_1)} \\
        &= \norm{\{K'(\theta_1) + K'(\theta_2) - K'(\theta_1)\}^TM(\theta_2)\{K'(\theta_1) + K'(\theta_2) - K'(\theta_1)\} -K'(\theta_1)^TM(\theta_1)K'(\theta_1)} \\
        &= \norm{K'(\theta_1)^T\{M(\theta_2) - M(\theta_1)\}K'(\theta_1)} + \norm{\sym\{K'(\theta_1)^TM(\theta_2)(K'(\theta_2) - K'(\theta_1))\}} \\
        &\hspace{5mm} + \norm{\{K'(\theta_2) - K'(\theta_1)\}^TM(\theta_2)\{K'(\theta_2) - K'(\theta_1)\}} \\
        &\leq \norm{K'(\theta_1)}^2\norm{M(\theta_2)-M(\theta_1)} + 2\norm{K'(\theta_1)}\norm{K'(\theta_2) - K'(\theta_1)}\norm{M(\theta_2)} \\
        &\quad + \norm{K'(\theta_2) - K'(\theta_1)}^2\norm{M(\theta_2)} \\
        &\leq 5e6\tau_{B(\theta_1)}^2\norm{P(\theta_1)}^{17}\norm{\theta_1-\theta_2}
    \end{align*}
    For the last fact, from \Cref{lem: simplifying inequalities} and \Cref{lem: Bound on K' and K''}, 
    \begin{align*}
        \norm{H(\theta_2)} &\leq \norm{K'(\theta_2)}^2\norm{M(\theta_2)} 
        \leq 2e3\norm{P(\theta_1)}^7\cdot4\tau^2_{B(\theta_1)}\norm{P(\theta_1)}^2 = 8e3\tau^2_{B(\theta_1)}\norm{P(\theta_1)}^9
    \end{align*}
\end{proof}

\section{Least Squares Analysis}
\label{s: id bound proof}

In this section, we provide a characterization of the least squares error for the procedure discussed in \Cref{s: methods}. The following result characterizes the weighted parameter identification error of the least squares. 

\begin{lemma}[Least Squares Identification Bound]
    \label{thm: identification bound}
    Let $H$ be a positive definite matrix belonging to $\mathbb{R}^{d_{\theta}\times d_{\theta}}$. Suppose the dataset $\curly{(X_t^n, U_t^n, X_{t+1}^n)}_{t=1, n=1}^{T,N}$ is collected from system \eqref{eq: linear system} using random noise $U_t \sim \calN(0, \Sigma_u)$. Let $\hat \theta$ be the least squares estimate \eqref{eq: least squares}, and $\hat {\mathsf{FI}}$ be the Fisher Information estimate \eqref{eq: fisher estimate}. Let $\delta \in (0,1)$. It holds that with probability at least $1-\delta$ that 
    \begin{align}
        \norm{\hat \theta - \theta^\star}_H^2 \leq 4 \frac{\trace\paren{H \mathsf{FI}(\theta^\star)^{-1}}}{N} + 8 \frac{\norm{H \mathsf{FI}(\theta^\star)^{-1}}}{N} \log\frac{2}{\delta},
    \end{align}
    and
        $0.5 \mathsf{FI}(\theta^\star) \preceq \hat{\mathsf{FI}} \preceq 2 \mathsf{FI}(\theta^\star)$
    as long as the number of experiments satisfies 
        $N \geq N_{\mathsf{ID}},$
    for $N_{\mathsf{ID}}$ in \eqref{eq: id burn-in}.  
\end{lemma}
The result follows standard arguments applying concentration inequalities to linear systems \cite{ziemann2023tutorial, tu2024learning}. The specific form with the weighting matrix $H$ is an instantiation of Theorem 3.1 in \cite{lee2024active}. For completeness, a proof is provided below. 

To prove \Cref{thm: identification bound}, we first state two supporting lemmas. The first is a standard result on covariance concentration from \citet{jedra2020finite}. We define the matrix $\Gamma_x$ as 
\begin{align*}
    \Gamma_x = \bmat{0 \\ \bmat{I & B} \\ A \bmat{I & B} &\bmat{I & B} \\ \vdots \\ A^{T-1} \bmat{I & B} & \dots \bmat{I & B}}
\end{align*}
such that for any experiment $n$,
\begin{align*}
    \bmat{X_1^n \\ \vdots \\ X_T^n} = \Gamma_x \bmat{W_1^n \\ U_1^n \\ \vdots \\ W_{T-1}^n \\ U_{T-1}^n}. 
\end{align*}

\begin{lemma}[Covariance Concentration]
    \label{lem: covariance concentration}
    Consider collecting $N$ trajectories of length $T$ from \eqref{eq: linear system}. 
    Let $\beta \in \paren{0, \frac{1}{4}\lambda_{\min}\paren{\mathbf{E} \sum_{t=1}^T \bmat{X_t \\ U_t} \bmat{X_t\\ U_t}^\top}}$. Define the event 
    \begin{align}
        \label{eq: covariance concentration event}
        \calE \triangleq \norm{ \frac{1}{N} \sum_{n=1}^N \sum_{t=1}^T \bmat{X_t^n \\ U_t^n} \bmat{X_t^n\\ U_t^n}^\top - \mathbf{E} \sum_{t=1}^T \bmat{X_t \\ U_t} \bmat{X_t\\ U_t}^\top} \leq \beta.
    \end{align}
    There exists a universal positive constant $c$ such that if 
    \begin{align*}
        N \geq c \frac{\norm{\Gamma_x}^2 \norm{\Sigma_u} \norm{\Sigma_w} \norm{\mathbf{E} \sum_{t=1}^T \bmat{X_t \\U_t}\bmat{X_t \\U_t}^\top}^2}{\lambda_{\min}\paren{\bfE  \sum_{t=1}^T \bmat{X_t \\U_t}\bmat{X_t \\U_t}^\top}\beta^2} \paren{\log\frac{1}{\delta} + \dx + \du},
    \end{align*}
    then $\calE$ holds with probability at least $1-\delta$.
\end{lemma}
\begin{proof}
Let  $\eta \sim \calN(0, I_{\dx \du T})$ and $\Gamma$ be the matrix mapping from 
\begin{align*}
    \bmat{W_1^n \\ U_1^n \\ \vdots \\ W_{T-1}^n \\ U_{T-1}} \textrm{ to } \bmat{X_1^n \\ U_1^n \\X_2^n \\ U_2^n \\ \vdots \\ X_T^n \\ U_T^n}. \mbox{ Then  } \bmat{X_1^n \\ U_1^n \\X_2^n \\ U_2^n \\ \vdots \\ X_T^n \\ U_T^n} \overset{d}{=} \tilde \Gamma \eta, \mbox{ where } \tilde \Gamma = \Gamma \paren{I_T \otimes \bmat{\Sigma_w^{1/2} \\ & \Sigma_u^{1/2}}}.
\end{align*}
Let $M = \left(N \mathbf{E} \sum_{t=1}^T \bmat{X_t \\ U_t} \bmat{X_t\\ U_t}^\top\right)^{-1/2}.$
It holds that
\begin{align*}
    &\norm{M \sum_{n=1}^N \sum_{t=1}^T \bmat{X_t^n \\ U_t^n} \bmat{X_t^n \\ U_t^n}^\top M - I}\\
    &= \sup_{v \in \calS^{d-1}} v^\top \paren{M \sum_{n=1}^N \sum_{t=1}^T \bmat{X_t^n \\ U_t^n} \bmat{X_t^n \\ U_t^n}^\top M - I}v \\
    &= \sup_{v \in \calS^{d-1}} v^\top \paren{M \sum_{n=1}^N \sum_{t=1}^T \bmat{X_t^n \\ U_t^n} \bmat{X_t^n \\ U_t^n}^\top M - \mathbf{E}\brac{\paren{M \sum_{n=1}^N \sum_{t=1}^T \bmat{X_t^n \\ U_t^n} \bmat{X_t^n \\ U_t^n}^\top} M}}v \\
    &= \sup_{v \in \calS^{d-1}} \norm{\sigma_{Mv}^\top \tilde \Gamma \eta}^2 - \mathbf{E} \norm{\sigma_{Mv}^\top \tilde \Gamma \eta}^2,
\end{align*}
where $\sigma_{Mv} \triangleq I_{NT} \otimes (Mv).$ By the Hanson-Wright inequality applied to Gaussian quadratic forms (presented for subGaussian forms in Theorem 6.3.2 of \citep{vershynin2020high}, and specialized to Gaussians by \citet{laurent2000adaptive}) along with a covering argument, it holds that with probability at least $1-\delta$, 
\begin{align*}
    \sup_{v \in \calS^{d-1}} \norm{\sigma_{Mv}^\top \tilde \Gamma \eta}^2 - \mathbf{E} \norm{\sigma_{Mv}^\top \tilde \Gamma \eta}^2 \leq \frac{\beta}{\norm{\mathbf{E} \sum_{t=1}^T \bmat{X_t \\ U_t} \bmat{X_t\\ U_t}^\top}}
\end{align*}
with probability at least $1-\exp\paren{-c_1 N \beta^2 \frac{\lambda_{\min}\paren{\mathbf{E} \sum_{t=1}^T \bmat{X_t \\ U_t} \bmat{X_t\\ U_t}^\top}}{\norm{\tilde \Gamma}^2 \norm{\mathbf{E} \sum_{t=1}^T \bmat{X_t \\ U_t} \bmat{X_t\\ U_t}^\top}^2} + c_2 d}$. Inverting this and applying submultiplicativity to bound $\norm{\tilde \Gamma}$ concludes the statement.

\end{proof}

The second is a self-normalized martingale bound, adapted in Lemma A.8 of \citet{lee2024active} from the standard self-normalized martingale bound in Theorem 14.7 of \citet{pena2009self}.The specific form that we use is from Lemma A.8 of \citet{lee2024active}. 
\begin{lemma}[Lemma A.8 of \citet{lee2024active}]
\label{lem: sn martingale bnd}
    Let $\curly{W_k}_{k=1}^K$ be a sequence of standard normal Gaussian random variables.  Let $\curly{Z_{k}}_{k=1}^K$ be a sequence of Gaussian random vectors assuming values in $\R^{d_{\theta}}$ such that $Z_k$ is independent from $W_\tau$ for $\tau \geq k$.  Let $\Sigma_Z \triangleq \mathbf{E} \frac{1}{K} \sum_{k=1}^K Z_k Z_k^\top$. Suppose $H \in \R^{d_{\theta} \times d_{\theta}}$ is positive definite and $\beta \in \R$ satisfies 
    \[
        0 < \beta \leq \frac{\lambda_{\min}\paren{\Sigma_Z}}{2}.
    \]
    As long as the event $\norm{\frac{1}{K} \sum_{k=1}^K Z_k Z_k^\top - \frac{1}{K} \mathbf{E} \sum_{k=1}^K Z_k Z_k^\top} \leq \beta$, then the following holds with probability at least $1-\delta$, 
    \begin{align*}
        &\norm{\left(\sum_{k=1}^K Z_k Z_k^\top\right)^{-1} \sum_{k=1}^K Z_k W_k}_H^2 \leq 2 \paren{1 + \frac{4\beta}{\lambda_{\min}(\Sigma_Z)}}\paren{ \trace\paren{(K \Sigma_Z)^{-1} H} + 2  \norm{(K \Sigma_Z)^{-1} H} \log \frac{1}{\delta}}. 
    \end{align*}
\end{lemma}
Leveraging the above two results, we can prove \Cref{thm: identification bound}. We write the least squares identification error as
\begin{align*}
    \bmat{ A(\hat \theta) & B(\hat \theta)} - \bmat{ A(\theta^\star) & B(\theta^\star)} = \sum_{t=1,n=1}^{T,N} W_{t}^n \bmat{X_t^n \\ U_t^n}^\top \left(\sum_{t=1,n=1}^{T,N} \bmat{X_t^n \\ U_t^n} \bmat{X_t^n \\ U_t^n}^\top\right).
\end{align*}
The noise covariance of $W_t^n$ can be pulled out such that for $\xi_{t,n} =\Sigma_w^{-1/2} W_{t}^n$, 
\begin{align*}
    \bmat{ A(\hat \theta) & B(\hat \theta)} - \bmat{ A(\theta^\star) & B(\theta^\star)} = \sum_{t=1,n=1}^{T,N} \Sigma_w^{1/2} \xi_{t}^n \bmat{X_t^n \\ U_t^n}^\top \left(\sum_{t=1,n=1}^{T,N} \bmat{X_t^n \\ U_t^n} \bmat{X_t^n \\ U_t^n}^\top\right).
\end{align*}
Applying the vectorization identity $\VEC(XYZ) = (Z^\top \otimes X) \VEC Y$, we find that
\begin{align*}
    &\VEC \paren{\bmat{ A(\hat \theta) & B(\hat \theta)} - \bmat{ A(\theta^\star) & B(\theta^\star)}} \\
    &= \left(\sum_{t=1,n=1}^{T,N}    \paren{\bmat{X_t^n \\ U_t^n} \otimes \Sigma_w^{-1/2}}\paren{ \bmat{X_t^n \\ U_t^n}^\top  \otimes \Sigma_w^{-1/2}}\right)^{-1}\sum_{t=1,n=1}^{T,N} \paren{\bmat{X_t^n \\ U_t^n} \otimes \Sigma_w^{-1/2}} \eta_t^n.
\end{align*}
This can be further decomposed by letting $Z_{t}^n[i]$ be the $i^{\mathsf{th}}$ column of $\bmat{X_t^n \\ U_t^n} \otimes \Sigma_w^{-1/2}$, and $\eta_t^n[i]$ be the $i^{\mathsf{th}}$ entry of $\eta_t^n$. Then 
\begin{align*}
    &\VEC \paren{\bmat{ A(\hat \theta) & B(\hat \theta)} - \bmat{ A(\theta^\star) & B(\theta^\star)}} = \left(\sum_{t=1, n=1, i=1}^{T,N,\dx} Z_t^n[i] Z_t^n[i]^\top\right)^{-1}\sum_{t=1,n=1,i=1}^{T,N, \dx} Z_t^n[i] \eta_t^n[i].
\end{align*}
Invoking \Cref{lem: covariance concentration}, the event
\begin{align*}
    \calE = \curly{\norm{\frac{1}{N} \sum_{t=1, n=1, i=1}^{T,N,\dx} Z_t^n[i] Z_t^n[i]^\top - \Sigma_Z} \leq \frac{1}{4} \lambda_{\min}(\Sigma_Z)},
\end{align*}
with $\Sigma_Z = \mathbf{E} \sum_{t=1, i=1}^{T,\dx} Z_t[i] Z_t[i]^\top$
holds with probability at least $1-\delta/2$ as long as 
\begin{align*}
   N \geq c \frac{\norm{\Gamma_x}^2 \norm{\Sigma_u} \norm{\Sigma_w}^3 \norm{\mathbf{E} \sum_{t=1}^T \bmat{X_t \\U_t}\bmat{X_t \\U_t}^\top}^2}{\lambda_{\min}\paren{\bfE  \sum_{t=1}^T \bmat{X_t \\U_t}\bmat{X_t \\U_t}^\top}^3 \lambda_{\min}(\Sigma_w)^2} \paren{\log\frac{1}{\delta} + \dx + \du},
\end{align*}
for a universal positive constant $c$.  Under this event, \Cref{lem: sn martingale bnd} implies that with probability at least $1-\delta/2$,
\begin{align*}
    & \norm{\VEC \paren{\bmat{ A(\hat \theta) & B(\hat \theta)} - \bmat{ A(\theta^\star) & B(\theta^\star)}}}_H^2 \leq  4\paren{ \trace\paren{(N \Sigma_Z)^{-1} H} + 2  \norm{(N \Sigma_Z)^{-1} H} \log \frac{2}{\delta}}.
\end{align*}
To conclude the proof, note that $\Sigma_Z = \mathsf{FI}(\theta^\star)$, and union bound over the success events. Additionally, note that $\norm{\mathbf{E} \sum_{t=1}^T \bmat{X_t \\U_t}\bmat{X_t \\U_t}^\top}$ and $\norm{\Gamma_x}$ may be bounded in terms of $\calJ(\theta^\star) = \sum_{t=0}^{T-1} \norm{A(\theta^\star)^t \bmat{I & B(\theta^\star)}}$, $\norm{\Sigma_w}$ and $\norm{\Sigma_w}$ \citep{jedra2020finite}. 

\section{Fundamental Limits}
\label{s: lower bound proof}

While good algorithm design can decrease the suboptimality gap of the learned controller, there are fundamental limits on the achievable performance. These limits are characterized by the signal-to-noise ratio of the experiment procedure, as well as the sensitivity of the LQR problem to error in the parameter estimates. To capture the signal-to-noise ratio for the experimental procedure, we define the Fisher Information matrix:
\begin{align}
    \mathsf{FI}(\theta) \triangleq \mathbf{E}_{\theta}\brac{\sum_{t=1}^T \bmat{X_t \\ U_t} \bmat{X_t \\ U_t}^\top } \otimes \Sigma_W^{-1}. \label{eq:FI}
\end{align}
To characterize the sensitivity of the LQR problem, we define the matrix $H(\theta)$ as 
\begin{align}
    \label{eq: model task Hessian}
    H(\theta) \triangleq D_{\theta} \VEC K(\theta)^\top  (\Sigma_X(\theta) \otimes (B(\theta)^\top P(\theta) B(\theta) + R)) D_{\theta} \VEC K(\theta).
\end{align}
One can verify that this matrix is the Hessian of the cost $C(K(\tilde \theta), \theta)$ with repsect to $\tilde \theta$, and evaluated at $\tilde \theta = \theta$ \citep{wagenmaker2021task}. The below theorem presents a lower bound on the  $\varepsilon$-local minimax excess cost gap in terms of these quantities. 
\begin{theorem}
    \label{thm: lower bound}
    Let $\rho(\theta', T, N)$ denote the distribution of the experiment data induced by running the aforementioned experiment procedure on the system $X_{t+1} = A(\theta') X_t + B(\theta') U_t + W_t$ for $N$ episodes of length $T$. Assume that $\rho(A(\theta^\star)) < 1$.   Consider applying any learning algorithm $\calA$ that maps a dataset \eqref{eq: dataset} to a controller $K$. Let $\varepsilon \in \R$ satisfy $0 \leq \varepsilon \leq {\varepsilon_{UB}}$, where $\varepsilon_{\mathsf{UB}} = \frac{1}{\mathsf{poly}(\norm{P(\theta^\star}, \tau_{B(\theta^\star)}, \norm{\mathsf{FI}(\theta^\star)}, \frac{1}{\lambda_{\min}(\mathsf{FI}(\theta^\star))})}$. Additionally suppose that $N \geq N_{LB}$ where $N_{LB} = \frac{1}{\varepsilon^2} \frac{1}{\lambda_{\min}(\mathsf{FI}(\theta^\star))} \mathsf{poly}(d_{\theta}, \norm{P(\theta)})$. It holds that 
    \begin{align*}
        &\sup_{\theta' \in \calB(\theta^\star, \varepsilon)} \mathbf{E}_{\mathsf{Data} \sim p(\theta', T, N)}\brac{C\paren{\calA\paren{\mathsf{Data}}, \theta'} - C(K(\theta'), \theta')} \geq \frac{1}{8} \trace\paren{H(\theta^\star) (N\mathsf{FI}(\theta^\star))^{-1}}.  
    \end{align*}
\end{theorem}
The above result follows from Theorem 2.2 of \citet{lee2023fundamental}; however, it is rewritten to demonstrate tight dependence on the system-theoretic quantities. A proof of this result is provided in below.

\begin{proof}

Denote the data by $Z$. Let $\lambda$ be a prior density over $\theta$ satisfying $\lambda(\theta) \propto (1 - \frac{1}{\varepsilon^2} \norm{\theta - \theta^\star}^2)^2$ for $\theta \in \calB(\theta^\star, \varepsilon)$. We may lower bound the minimax quantity by an expectation over the prior:
\begin{align*}
    &\sup_{\theta'\in\calB(\theta^\star, \varepsilon)} \bfE_{\mathsf{Z} \sim p(\theta', T, N)} \brac{C(\calA(Z), \theta') -C(K(\theta'), \theta') } \geq \bfE_{\Theta \sim \lambda} \bfE_{Z \sim p(\Theta, T, N)} \brac{C(\calA(Z), \Theta) -C(K(\Theta), \Theta) }.
    \end{align*}
Next, we will apply the performance difference lemma to lower bound the expected excess cost in terms of the gap between the controller output by our algorithm and the optimal controller. To do so, we condition on the event that our algorithm outputs a controller close enough to the certainty equivalent controllers within the ball. In particular, we define 
\begin{align*}
    \calE &= \curly{\sup_{\theta\in\calB(\theta^\star, \varepsilon)} \norm{\calA(Z) - K(\theta)} \leq \alpha }, \mbox{ where }
    \alpha = \inf_{\theta\in\calB(\theta^\star, \varepsilon)} \frac{1}{12 \norm{\Sigma^{K(\theta)}}^{5/2}}.
\end{align*}
Additionally define $\tilde \Psi$ and $\tilde \Sigma$ as the largest matrices in semidefinite order such that $\tilde \Sigma \preceq \Sigma^{K(\theta)}(\theta)$ and $\tilde \Psi \preceq \Psi(\theta)$ for all $\theta \in B(\theta^\star, \varepsilon)$. 
This allows us to apply \Cref{lem: performance difference} to achieve the following lower bound:
\begin{align*}
    &\sup_{\theta'\in\calB(\theta^\star, \varepsilon)} \bfE_{\mathsf{Z} \sim p(\theta', T, N)} \brac{C(\calA(Z), \theta') -C(K(\theta'), \theta') } \\
    &\geq  \bfE_{\Theta \sim \lambda} \bfE_{Z\sim p(\Theta, T, N)} \brac{ \trace\paren{(\calA(Z)  - K(\Theta)) \Sigma^{\calA(Z)}(\Theta) (\calA(Z) - K(\Theta))^\top \Psi(\Theta)} \mathbf{1}_{\calE}}, \\
    &\geq \frac{1}{2} \bfE_{\Theta \sim \lambda} \bfE_{Z \sim p(\Theta, T, N)} \brac{ \trace\paren{(\calA(Z) - K(\Theta)) \tilde \Sigma (\calA(Z) - K(\Theta))^\top \tilde \Psi} \mathbf{1}_{\calE}},
\end{align*}
where the final inequality follows from the fact that if $\calE$ holds, then \Cref{lem: lyap perturbation} ensures $\Sigma^{\calA(\mathsf{Data}}(\Theta) \succeq \frac{1}{2}\Sigma^{K(\Theta)}(\Theta),$ (\Cref{lem: cov lower bound}) and by substituting the lower bounds $\tilde \Sigma \preceq \Sigma^{K(\Theta)}(\Theta)$, $\tilde \Psi \preceq \Psi(\Theta)$. 

By application of the Van Trees inequality, as in Theorem 2.1 of \citet{lee2023fundamental}, it holds that 
\begin{align*}
    &\sup_{\theta'\in\calB(\theta^\star, \varepsilon)} \bfE_{Z \sim p(\theta', T, N)} \brac{C(\calA(Z), \theta') -C(K(\theta'), \theta') } \\
    &\geq \frac{1}{2} \trace\paren{ \paren{\tilde \Sigma \otimes \tilde \Psi}  \bfE \brac{D_{\theta} \VEC K(\Theta) \mathsf{1}_{\calE}} \paren{N\bfE\brac{\mathsf{FI}(\Theta)} + J(\lambda)}^{-1} \bfE \brac{D_{\theta} \VEC K(\Theta) \mathsf{1}_{\calE}} } \\
    &\geq \frac{1}{2} \inf_{\tilde \theta_1, \tilde \theta_2, \tilde\theta_3 \in \calB(\theta^\star, \varepsilon)} \trace\paren{ \paren{\tilde \Sigma \otimes \tilde \Psi}  D_{\theta} \VEC K(\tilde \theta_1)  \paren{N\mathsf{FI}(\tilde \theta_3) + J(\lambda)}^{-1} D_{\theta} \VEC K(\tilde \theta_2)}  \mathbf{P}(\calE)^2,
\end{align*}
where $J(\lambda) = \int \nabla \log \lambda(\theta) (\nabla \log \lambda(\theta))^\top \lambda(\theta) d\theta$ satisfies $\norm{J(\lambda)} \leq \frac{1}{\varepsilon^2} \frac{32}{d_\theta+2} \frac{\Gamma((d_\theta+5)/2)}{\Gamma(d_{\theta}/2)^2}$ (by the triangle inequality and direct calculation). Furthermore, first order Taylor expansions of $\Sigma^{K(\theta)}(\theta)\kron \Psi(\theta)$, $D_\theta \mathsf{vec} K(\theta)$, and $\mathsf{FI}(\theta)$ about $\theta^\star$, combined with the bounds of \Cref{lem: LQR Taylor expansion}, \Cref{lem: helper lemma for RC}, and \Cref{lem: lyap perturbation} we can express the above quantity as
\begin{align}\label{eq: lb perturbation}
    \frac{1}{2}  \trace\paren{ \paren{H(\theta^\star) + M_1}  (D_{\theta} \VEC K( \theta^\star) + M_2)  \paren{N\mathsf{FI}(\theta^\star) +  M_3+ J(\lambda)}^{-1} (D_{\theta} \VEC K(\theta^\star) + M_4)}  \mathbf{P}(\calE)^2,
\end{align} 
\sloppy where $\norm{M_1} \leq 1e7 \tau_{B(\theta^\star)}^2 \norm{P(\theta^\star)}^{15} \varepsilon$, $\norm{M_2} \leq 2000 \norm{P(\theta^\star)}^{15/2} \varepsilon$, $\norm{M_3} \leq  2N \norm{\mathsf{FI}(\theta^\star)} \varepsilon$ and $\norm{M_4} \leq 2000 \norm{P(\theta^\star)}^{15/2} \varepsilon$.

We will show that the above quantity is at least 
\begin{align*}
    \frac{1}{8}  \trace\paren{ \paren{H(\theta^\star)}  (D_{\theta} \VEC K( \theta^\star))  \paren{N\mathsf{FI}(\theta^\star)}^{-1} (D_{\theta} \VEC K(\theta^\star))}.
\end{align*}
To do so, assume to the contrary that the inequality does not hold. 
Under this assumption, we will show that $\bfP(\calE) \geq \frac{1}{\sqrt{2}}$. This follows by observing that $\sup_{\theta\in\calB(\theta^\star, \varepsilon)} \norm{\calA(Z) - K(\theta)}\leq \inf_{\theta\in B(\theta^\star, \varepsilon)} \norm{\calA(Z) - K(\theta)} + \sup_{\theta_1,\theta_2\in\calB(\theta^\star, \varepsilon))} \norm{K(\theta_1) - K(\theta_2)} \leq\inf_{\theta\in B(\theta^\star, \varepsilon)} \norm{\calA(Z) - K(\theta)}  + 64 \norm{P(\theta^\star)}^{7/2} \varepsilon$. By a sufficiently small choice of $\varepsilon$, we may show the desired condition holds if  $\inf_{\theta\in B(\theta^\star, \varepsilon)} \norm{\calA(Z) - K(\theta)} \leq \alpha/2$ with high probability. In particular, we may bound $\inf_{\theta\in B(\theta^\star, \varepsilon)} \norm{\calA(Z) - K(\theta)} \leq \sqrt{{\norm{\calA(Z) - K(\Theta)}^2}} \leq \sqrt{\trace((\calA(Z) - K(\Theta)) \Sigma^{\calA(Z)}(\Theta)(\calA(Z) - K(\Theta))^\top \Psi(\Theta)}$. This quantity is precisely the excess cost of the algorithm applied to dataset $Z$ on system $\Theta$. Markov's inequality then implies that this quantity exceeds $\alpha^2/4$ with probabiliy at most $\frac{4\bfE\brac{C(\calA(Z), \Theta) - C(K(\Theta), \Theta)}}{\alpha^2}$. Under our assumption, this probability can be bounded as
\begin{align*}
    P(\calE^c) \leq \frac{1}{2 \alpha^2}  \trace\paren{ \paren{H(\theta^\star)}  (D_{\theta} \VEC K( \theta^\star))  \paren{N\mathsf{FI}(\theta^\star)}^{-1} (D_{\theta} \VEC K(\theta^\star))}.
\end{align*}
For $N$ sufficiently large, as given in the statement, this implies that $P(\calE) \geq \frac{1}{\sqrt{2}}$. Then, by a sufficiently small choice of $\varepsilon$ as given in the theorem statement, \eqref{eq: lb perturbation} implies that the minimax excess cost exceeds $\frac{1}{8}  \trace\paren{ \paren{H(\theta^\star)}  (D_{\theta} \VEC K( \theta^\star))  \paren{N\mathsf{FI}(\theta^\star)}^{-1} (D_{\theta} \VEC K(\theta^\star))}$, contradicting the assumption that the excess cost falls below this quantity. 

\end{proof}

\section{Certainty Equivalence Upper Bound}
\label{s: certainty equivalence bound}
Prior work \citep{wagenmaker2021task} has demonstrated that certainty equivalence is efficient by providing upper bounds on the excess cost that match the lower bound of \Cref{thm: lower bound}. In particular, the following bound holds.
\begin{theorem}
    \label{thm: certainty equivalence bound}
    Suppose the dataset $\curly{(X_t^n, U_t^n, X_{t+1}^n)}_{t=1, n=1}^{T,N}$ is collected from N trajectories of the system \eqref{eq: linear system} via a random control input $U_t \sim \calN(0, \Sigma_u)$. Let $\hat\theta$ be the least square estimate computed by \eqref{eq: least squares}. Let also $\delta\in(0,1)$. Then it holds with probability at least $1-\delta$ that 
    \begin{align}
        &C(K_{CE}(\hat\theta), \theta^\star) - C(K(\theta^\star), \theta^\star) \nonumber\\
        &\leq 4\frac{\trace\paren{H(\theta^\star)\mathsf{FI}(\theta^\star)^{-1}}}{N} + 8\frac{\norm{H(\theta^\star)\mathsf{FI}(\theta^\star)^{-1}}}{N}\log\frac{2}{\delta} + L_{\mathsf{CE}}(\theta^\star)\frac{\norm{\mathsf{FI}(\theta^\star)^{-1}}^{3/2}}{N^{3/2}}, \label{eq:CE Upper bound}
    \end{align}
    where $L_{\mathsf{CE}}(\theta^\star) = 2e7\tau_{B(\theta^\star)}^3\norm{P(\theta^\star)}^{14}\paren{d_\theta+\log\frac{2}{\delta}}^{3/2}$, as long as the number of trajectories $N$ satisfies
     $   N \geq \max\curly{N_{\mathsf{ID}}, ~ 6e5(d_\theta+\log\frac{2}{\delta})\norm{\mathsf{FI}(\theta^\star)^{-1}}\norm{P(\theta^\star)}^{10}}.$
\end{theorem}
The above result is sharpened from Theorem 2.1 of \citet{wagenmaker2021task} to avoid a logarthmic factor of the state dimension. A proof is provided in \Cref{subsec: Proof of CE upper bound}. By inverting the above high probability tail bound, one can show that for $N$ sufficiently large, the following inequality holds:
\begin{align*}
    \mathbf{E}_{\mathsf{data}} \brac{C(K_{\mathsf{CE}}(\hat\theta), \theta^\star) - C(K(\theta^\star), \theta^\star)} \leq c \trace\paren{H(\theta^\star) \mathsf{FI}(\theta^\star)^{-1}},
\end{align*}
where $c$ is a universal positive constant. This matches the lower bound of \Cref{thm: lower bound} up to a universal constant. 

We first present a helping lemma that bounds the suboptimality gap in terms of a quadratic function of the parameter estimation error. 

\begin{lemma}
    \label{lem: CE upper bound}
    Suppose $\hat\theta$ is some parameter satisfying $\norm{\hat\theta-\theta^\star} \leq \frac{1}{256}\norm{P(\theta^\star)}^{-5}$. Then the excess cost of $K_{CE}(\hat\theta)$ would be
    \begin{align*}
        C(K_{CE}(\hat\theta), \theta^\star) - C(K(\theta^\star), \theta^\star) \leq \|\hat\theta - \theta^\star\|_{H(\theta^\star)}^2 +   L_{\mathsf{RC}}(\theta^\star)\norm{\hat\theta-\theta^\star}^3, 
    \end{align*}
    where 
    \begin{align*}
        L_{\mathsf{RC}}(\theta^\star) = 6e5\tau_{B(\theta^\star)}^3\norm{P(\theta^\star)}^{14}
    \end{align*}
\end{lemma}
\begin{proof}
    When $\|\hat\theta - \theta^\star\|\leq\frac{1}{16}\norm{P(\theta^\star)}^{-2}\leq \frac{1}{256} \norm{P(\theta^\star)}^{-5}$, 
    from \Cref{lem: Riccati perturbation}, \Cref{lem: CE stabilization}, \Cref{lem: excess cost decomposition} and \Cref{lem: cost gap taylor substitution}, we have
    \begin{align*}
        &C(K(\hat\theta), \theta^\star) - C(K(\theta^\star), \theta^\star)\\ 
        &\leq \|\hat\theta - \theta^\star\|_{H(\theta^\star)}^2 + 2 \norm{P(\theta^\star)}^2 \|\Sigma_X^{K(\hat\theta)}(\theta^\star)\|\tau_{B(\theta^\star)}^3\|(K(\hat{\theta})-K(\theta^\star))\|^3(\|B(\theta^\star)(K(\hat\theta-K(\theta^\star))\| + 2 \norm{P(\theta^\star)}^{1/2}) \\
        &\hspace{5mm}+ 2e5\tau_{B(\theta^\star)}^2\|P(\theta^\star)\|^{13}\|\hat\theta-\theta^\star\|^3 
        + 8e6\tau_{B(\theta^\star)}^2\|P(\theta^\star
        )\|^{17}\|\|\hat\theta-\theta^\star\|^4 \quad \mbox{(\Cref{lem: excess cost decomposition} and \Cref{lem: cost gap taylor substitution})} \\
        &\leq \|\hat\theta - \theta^\star\|_{H(\theta^\star)}^2 + 2^{17} \norm{P(\theta^\star)}^3\tau_{B(\theta^\star)}^3
        \norm{P(\theta^\star)}^{21/2} \norm{\hat\theta-\theta^\star}^3
        (32\norm{P(\theta^\star)}^{7/2} \norm{\hat\theta-\theta^\star} + 2 \norm{P(\theta^\star)}^{1/2}) \\
        &\hspace{5mm}+ 2e5\tau_{B(\theta^\star)}^2\|P(\theta^\star)\|^{13}\|\hat\theta-\theta^\star\|^3 
        + 8e6\tau_{B(\theta^\star)}^2\|P(\theta^\star
        )\|^{17}\|\|\hat\theta-\theta^\star\|^4 \quad \mbox{(\Cref{lem: CE stabilization} and \Cref{lem: Riccati perturbation})}\\
        &\leq \|\hat\theta - \theta^\star\|_{H(\theta^\star)}^2 + 
        \left(
        2^{18}\tau_{B(\theta^\star)}^3 + 2e5\tau_{B(\theta^\star)}^2
        \right)\norm{P(\theta^\star)}^{14}\norm{\hat\theta-\theta^\star}^3 \\
        &\hspace{5mm}+ \left(
        2^{22}\tau_{B(\theta^\star)}^3 + 8e6\tau_{B(\theta^\star)}^2
        \right)\norm{P(\theta^\star)}^{17}\norm{\hat\theta-\theta^\star}^4\\
        &\leq  \|\hat\theta - \theta^\star\|_{H(\theta^\star)}^2 +     6e5\tau_{B(\theta^\star)}^3\norm{P(\theta^\star)}^{14}\norm{\hat\theta-\theta^\star}^3, 
    \end{align*}
    where the final inequality follows from the fact that $\norm{\hat\theta-\theta^\star}\leq \frac{1}{256}\norm{P(\theta^\star)}^{-5}$.
\end{proof}

\subsection{Proof of \texorpdfstring{\Cref{thm: certainty equivalence bound}}{}}
\label{subsec: Proof of CE upper bound}
\begin{proof}
    From \Cref{lem: CE upper bound} and \Cref{thm: identification bound}, we get \eqref{eq:CE Upper bound} where we used $\trace\paren{\mathsf{FI}(\theta^\star)^{-1}}\leq\norm{\mathsf{FI}(\theta^\star)^{-1}}d_\theta$. Furthermore, from the closeness condition in \Cref{lem: CE upper bound}, i.e.  $\norm{\hat\theta-\theta^\star}\leq\frac{1}{256}\norm{P(\theta^\star)}^{-5}$, $N$ needs to satisfy the following condition:
    \begin{align*}
        &\norm{\hat\theta-\theta^\star}^2 \leq \frac{8(d_\theta+\log\frac{2}{\delta})\norm{\mathsf{FI}(\theta^\star)^{-1}}}{N}\leq\frac{1}{256^2}\norm{P(\theta^\star)}^{-10} \\
        &\iff N \geq 6e5(d_\theta+\log\frac{2}{\delta})\norm{\mathsf{FI}(\theta^\star)^{-1}}\norm{P(\theta^\star)}^{10}.
    \end{align*}
\end{proof}

\section{Proof of Robust Control Upper Bound}
\label{s: robust control proof}

\begin{lemma}
    \label{lem: Robust Control Upper Bound}
    Let $G$ denote the ellipsoid:
    \begin{align*}
        G \triangleq \{\theta ~:~ \theta=\hat \theta + \Gamma w, ~ w\in\mathcal{B}(0, 1)\}.
    \end{align*}
    Suppose $G$ is $R$-robustly stabilizable, the diameter of $G$ satisfies $\mathsf{diam}(G)\leq\frac{1}{16}\inf_{\theta\in G}\norm{P(\theta)}^{-2}$, and $\theta^\star\in G$. Then the excess cost of $K_{RC}(G)$ would be
    \begin{align*}
        C(K(G), \theta^\star) - C(K(\theta^\star), \theta^\star) \leq \sup_{\theta_1, \theta_2\in G}\paren{\norm{\theta_1-\theta_2}_{H(\theta^\star)} + L_{RC}(\theta^\star)\norm{\theta_1-\theta_2}^3}, 
    \end{align*}
    where
    \begin{align*}
        L_{\mathsf{RC}}(\theta^\star) = 6e6M\tau^3_{B(\theta^\star)}\norm{P(\theta^\star)}^{17}.
    \end{align*}
\end{lemma}

\begin{proof}
    By the definition of the robust controller and \Cref{lem: excess cost decomposition},
    \begin{align*}
        &C(K(G), \theta^\star) - C(K(\theta^\star), \theta^\star) \\
        &\leq \sup_{\theta\in G}C(K(G), \theta) - C(K(\theta), \theta) = \inf_K\sup_{\theta\in G}C(K, \theta) - C(K(\theta), \theta) \\
        &\leq \inf_K\sup_{\theta\in G} \trace ((K-K(\theta))\Sigma_X^{K(\theta)}(\theta)(K-K(\theta))^T\Psi(\theta))\\
        &\hspace{5mm}+ 2 \norm{P(\theta)}^2 \norm{\Sigma_X^{K}(\theta)}\tau_{B(\theta)}^3\norm{K-K(\theta)}^3(\norm{B(\theta)(K-K(\theta))} + 2 \norm{P(\theta)}^{1/2}) \\
        &\triangleq \inf_K\sup_{\theta\in G}\mathsf{UB}(K, \Sigma^K_{\theta}, \theta),
    \end{align*}
    where we defined $\mathsf{UB}(K, \Sigma^K_{\theta}, \theta)$ as in the last equality.  
    Now restricting our policy $K$ to CE policies $K(\tilde\theta)$ yields
    \begin{align*}
        C(K(G), \theta^\star) - C(K(\theta^\star, \theta^\star) 
        &\leq \inf_{\tilde\theta\in G}\sup_{\theta\in G}\mathsf{UB}(K(\tilde\theta), \Sigma^{K(\tilde\theta)}(\theta), \theta) \\
        &\leq \inf_{\tilde\theta_1\in G}\sup_{\tilde\theta_2, \theta\in G}\mathsf{UB}(K(\tilde\theta_2), \Sigma^{K(\tilde\theta_1)}(\theta), \theta), 
    \end{align*}
    where we took infimum only over $K(\tilde\theta_1)$ in the last inequality. 
    Then since $\mathsf{diam}(G)\leq\frac{1}{16}\inf_{\theta\in G}\norm{P(\theta)}^{-2}$, $\norm{\tilde\theta_2-\theta}\leq\frac{1}{16}\norm{P(\theta)}^{-2}$ for any $\tilde\theta_2, \theta\in G$. 
    Thus we can apply \Cref{lem: cost gap taylor substitution} and the first term of $\mathsf{UB}(K(\tilde\theta_2), \Sigma^{K(\tilde\theta_1)}(\theta), \theta)$ results in 
    \begin{align*}
        &\sup_{\tilde\theta_2, \theta\in G}\trace ((K(\tilde\theta_2)-K(\theta))\Sigma_X^{K(\theta)}(\theta)(K(\tilde\theta_2)-K(\theta))^T\Psi(\theta)) \\
        &\leq \sup_{\tilde\theta_2, \theta\in G}\norm{\tilde\theta_2-\theta}_{H(\theta)}^2 + 2e5\tau_{B(\theta)}^2\norm{P(\theta)}^{13}\norm{\tilde\theta_2-\theta}^3 + 8e6\tau_{B(\theta)}^2\norm{P(\theta)}^{17}\|\norm{\tilde\theta_2-\theta}^4 .
    \end{align*} 
    From \Cref{lem: helper lemma for RC}, 
    \begin{align*}
        \sup_{\tilde\theta_2, \theta\in G}\norm{\tilde\theta_2-\theta}_{H(\theta)}^2
        &= \sup_{\tilde\theta_2, \theta\in G}\paren{\tilde\theta_2-\theta}^TH(\theta)\paren{\tilde\theta_2-\theta} \\
        &\leq  \sup_{\tilde\theta_2, \theta\in G}\norm{\tilde\theta_2-\theta}_{H(\theta^\star)} + 5e6\tau_{B(\theta^\star)}^2\norm{P(\theta^\star)}^{17}\norm{\tilde\theta_2-\theta}^3.
    \end{align*}
    
    For the second term, from \Cref{lem: Riccati perturbation} and robust stabilizability of $G$, we get
    \begin{align*}
        &\inf_{\tilde\theta_1\in G}\sup_{\tilde\theta_2, \theta\in G}2 \norm{P(\theta)}^2 \norm{\Sigma_X^{K(\tilde\theta_1)}(\theta)}\tau_{B(\theta)}^3\norm{K(\tilde\theta_2)-K(\theta)}^3\left(\norm{B(\theta)(K(\tilde\theta_2)-K(\theta))} + 2 \norm{P(\theta)}^{1/2}\right) \\
        &\leq \sup_{\tilde\theta_2, \theta\in G}2\cdot32^3M\tau^3_{B(\theta)}\norm{P(\theta)}^{21/2}\norm{\tilde\theta_2-\theta}^3\left(32\norm{P(\theta)}^{7/2}\norm{\tilde\theta_2-\theta} + 2\norm{P(\theta)}^{1/2}\right) \\
        &\leq 2^{17}M\tau^3_{B(\theta)}\norm{P(\theta)}^{11}\norm{\tilde\theta_2-\theta}^3 + 2^{21}M\tau^3_{B(\theta)}\norm{P(\theta)}^{14}\norm{\tilde\theta_2-\theta}^4.
    \end{align*}
    By applying $\norm{\tilde\theta_2-\theta}\leq\frac{1}{16}\norm{P(\theta)}^{-2}$ and grouping the term, we get
    \begin{align*}
        &C(K(G), \theta^\star) - C(K(\theta^\star, \theta^\star) \leq \sup_{\tilde\theta_2, \theta\in G}\norm{\tilde\theta_2-\theta}_{H(\theta^\star)} + 6e6M\tau^3_{B(\theta^\star)}\norm{P(\theta^\star)}^{17}\norm{\tilde\theta_2-\theta}^3.
    \end{align*}
\end{proof}

\subsection{Proof of \texorpdfstring{\Cref{thm: Robust Control upper bound}}{}}
\label{subsec: proof of RC upper bound}
\begin{proof}
    From \Cref{thm: identification bound}, 
    \begin{align*}
        (\hat\theta-\theta^\star)^\top \hat{\mathsf{FI}}N (\hat\theta - \theta^\star) &\leq 8\trace\paren{\hat{\mathsf{FI}} \times \mathsf{FI}(\theta^\star)^{-1}} + 8\norm{\hat{\mathsf{FI}}\times\mathsf{FI}(\theta^\star)^{-1}}\log\frac{2}{\delta} \\
        &\leq 16\paren{d_\theta+\log\frac{2}{\delta}},
    \end{align*}
    where we used $\hat{\mathsf{FI}}\preceq2\mathsf{FI}(\theta^\star)$. Thus $\theta^\star\in G$ with probability at least $1-\delta$. Also, from the robust stabilizability condition and the diameter condition in \Cref{lem: Robust Control Upper Bound}, i.e., $\mathsf{diam}(G)\leq\frac{1}{16}\inf_{\theta\in G}\norm{P(\theta)}^{-2}$, $N$ must satisfy
    \begin{align*}
        &\sup_{\theta_1, \theta_2\in G}\norm{\theta_1-\theta_2}^2
        \leq \frac{32\paren{d_\theta+\log\frac{2}{\delta}}}{N\lambda_{\min}\paren{\hat{\mathsf{FI}}}} 
        \leq \min\curly{\frac{1}{256}\inf_{\theta\in G}\norm{P(\theta)}^{-4}, r^2} \\
        &\iff N \geq \max\curly{\frac{2e4\norm{P(\theta^\star)}^4\paren{d_\theta+\log\frac{2}{\delta}}}{\lambda_{\min}\paren{\mathsf{FI}(\theta^\star)}}, \frac{64 \paren{d_\theta+\log\frac{2}{\delta}}}{\mathsf{r}^2 \lambda_{\min}\paren{\mathsf{FI}(\theta^\star)}}}.
    \end{align*}
    where we applied $0.5\mathsf{FI}(\theta^\star)\preceq\hat{\mathsf{FI}}(\theta^\star)$ and \Cref{lem: Riccati perturbation} to get the last inequality. 
    Furthermore,it follows from \eqref{eq: confidence ellipsoid} that
    \begin{align*}
        \sup_{\theta_1, \theta_2\in G}\norm{\theta_1-\theta_2}_{H\paren{\theta^\star}} 
        &\leq \frac{64\paren{d_\theta+\log\frac{2}{\delta}}\norm{H(\theta^\star)\mathsf{FI}(\theta^\star)^{-1}}}{N}.
    \end{align*}
    which, combined with \Cref{lem: Robust Control Upper Bound}, result in the upper bound in \Cref{thm: Robust Control upper bound}. 
\end{proof}

\section{Proof of Domain Randomization Upper Bound}
\label{s: domain randomization proofs}

We begin our proof with an intermediate results that characterizes the average excess cost incurred by the domain randomized controller relative to the optimal costs of all systems in the randomization domain.  
 
\begin{lemma}
    \label{lem: DR objective upper bound}
    Let $G$ denote the ellipsoid:
    \begin{align*}
        G \triangleq \{\theta ~:~ \theta=\hat \theta + \Gamma w, ~ w\in\mathcal{B}(0, 1)\},
    \end{align*} 
    Suppose the diameter of $G$ satisfies $\mathsf{diam(G)}\leq\frac{1}{256}\inf_{\theta\in G}\norm{P(\theta)}^{-5}$ and $\theta^\star\in G$. 
    Let $\calD$ be a distribution of system parameters over $G$. Then it holds that
    \begin{align*}
        \mathbf{E}_{\theta \sim \calD} 
        \brac{C(K_{\mathsf{DR}}(\calD), \theta) - C(K(\theta), \theta)} \leq \trace\paren{\mathbf{V}(\calD) H(\theta^\star)} + L_{\mathsf{DR,1}}(\theta^\star)  \norm{\Gamma}^3,
    \end{align*}
    where 
    \begin{align*}
        L_{\mathsf{DR,1}}(\theta^\star) = 1.7e8\tau_{B(\theta^\star)}^3\norm{P(\theta^\star)}^{17}.
    \end{align*} 
\end{lemma}
\begin{proof}
    By the definition of the domain randomized controller
    \begin{align*}
        \mathbf{E}_{\theta \sim \calD} 
        \brac{C(K_{\mathsf{DR}}(\calD), \theta) - C(K(\theta), \theta)} &= \min_{K} \mathbf{E}_{\theta \sim \calD} 
        \brac{C(K , \theta) - C(K(\theta), \theta)} \\
        &= \min_{K} \mathbf{E}_{\theta \sim \calD} \trace\paren{(K-K(\theta))\Sigma^K(\theta) (K-K(\theta))^\top \Psi(\theta)}
    \end{align*}
    where the second equality holds from \Cref{lem: performance difference}. 
    Then we can plug in the certainty equivalence controller defined by $K(\hat \theta)$ to achieve an upper bound:
    \begin{align*}
        \mathbf{E}_{\theta \sim \calD} 
        \brac{C(K_{\mathsf{DR}}(\calD), \theta) - C(K(\theta), \theta)} &\leq \mathbf{E}_{\theta \sim \calD} \trace\paren{(K(\hat\theta)-K(\theta))\Sigma^{K(\hat\theta)}(\theta) (K(\hat\theta)-K(\theta))^\top \Psi(\theta)},
    \end{align*}
    where we observe that due to the restricted diameter of $G$, \Cref{lem: CE stabilization} ensures $K(\hat\theta)$ stabilizes all system instances in the support of the distribution. In particular, by substituting the bound of \Cref{lem: CE upper bound} into the above inequality, it holds that
    \begin{align}
        \mathbf{E}_{\theta \sim \calD} 
        \brac{C(K_{\mathsf{DR}}(\calD), \theta) - C(K(\theta), \theta)} \leq \mathbf{E}_{\theta \sim \calD}  \brac{\|\hat\theta - \theta\|_{H(\theta)}^2}  +   \mathbf{E}_{\theta \sim \calD}  \brac{6e5\tau_{B(\theta)}^3\norm{P(\theta)}^{14}\norm{\hat\theta-\theta}^3}. \label{eq: DR objective upper bound with theta}
    \end{align}
    The second term may be bounded by leveraging the radius of the support for $\calD$. In particular, it holds that for all $\theta\in G$, $\norm{\hat \theta - \theta}^3 \leq \norm{\Gamma}^3$. 
    From \Cref{lem: helper lemma for RC}, the first term can be bounded by
    \begin{align*}
        \norm{\hat\theta-\theta}_{H(\theta)} \leq \norm{\hat\theta-\theta}_{H(\theta^\star)} + 5e6\tau_{B(\theta^\star)}^2\norm{P(\theta^\star)}^{17}\norm{\Gamma}^3.
    \end{align*}
    Consequently the expectation evaluates to the trace of the variance of the distribution $\calD$, weighted by $H(\theta^\star)$. Combining these bounds leads to the inequality in the statement. 
\end{proof}

We now leverage the statement about the the domain randomization objective to characterize several system theoretic quantities for the system $\theta^\star$ under the the domain randomized controller. 
\begin{lemma}
    \label{lem: DR helper lemmas}
    Let $G$ denote the ellipsoid:
    \begin{align*}
        G \triangleq \{\theta ~:~ \theta=\hat \theta + \Gamma w, ~ w\in\mathcal{B}(0, 1)\}.
    \end{align*} 
    Suppose the diameter of $G$ satisfies $\mathsf{diam(G)}\leq\frac{1}{256}\inf_{\theta\in G}\norm{P(\theta)}^{-5}$ and $\theta^\star\in G$. 
    Let $\calD$ be a distribution of system parameters over $G$. Let 
    \begin{align*}
        \textnormal{DR Objective } \triangleq  \bfE_{\theta\sim\calD} \brac{\trace\paren{(K_{\mathsf{DR}}(\calD) - K(\theta)) \Sigma^{K_{\mathsf{DR}}(\calD)}(\theta)  (K_{\mathsf{DR}}(\calD) - K(\theta)) \Psi(\theta)}}.
    \end{align*}
    It holds that 
    \begin{itemize}
        \item $\bfE_{\theta\sim\calD} \norm{K_{\mathsf{DR}}(\calD)- K(\theta)}^2 \leq \textnormal{DR Objective}$
        \item $\norm{K_{\mathsf{DR}}(\calD) - K(\theta^\star)} \leq \textnormal{DR Objective} + 32 \norm{P(\theta^\star)}^{7/2} \norm{\Gamma}$
        \item $\norm{K_{\mathsf{DR}}(\calD)}\leq \textnormal{DR Objective} + 2\norm{P(\theta^\star)}^{1/2}$ 
        \item $\textnormal{DR Objective} \leq 2.6e3\tau^3_{B(\theta^\star)}\norm{P(\theta^\star)}^4 \norm{\Gamma}$
        \item $\norm{\Sigma^{K_{\mathsf{DR}}(\calD)}(\theta)}\leq2\norm{\Sigma^{K_{\mathsf{DR}}(\calD)}(\theta^\star)}$ if $\mathsf{diam}(G)\leq\frac{1}{8\norm{\Sigma^{K_{\mathsf{DR}}}(\calD)(\theta^\star)}^{3/2}\paren{1+\norm{K_{\mathsf{DR}}}}^2}$
        \item $\norm{\Sigma^{K_{\mathsf{DR}}(\calD)}(\theta^\star)} 
        \leq 2\norm{\Sigma^{K(\theta^\star)}(\theta^\star)}$ if $\mathsf{diam}(G)\leq\frac{1}{6.4e4\tau^8_{B(\theta^\star)}\norm{P(\theta^\star)}^4\norm{\Sigma^{K(\theta^\star)}(\theta^\star)}^{3/2}}$. 
    \end{itemize}
\end{lemma}
\begin{proof}
    The first fact follows immediately from \Cref{lem: performance difference} along with the fact that $\Sigma^{K_{\mathsf{DR}}(\calD)}(\theta) \succeq I$ and $\Psi(\theta) \succeq I$. The second fact is derived from the first fact and \Cref{lem: Riccati perturbation}, which also yields the third fact using the diameter condition of $G$. 
    DR Objective is bounded as
    \begin{align*}
        \textnormal{DR Objective} &\leq \sup_{\theta\in G}\norm{\hat\theta-\theta}_{H(\theta)}^2 + 1.6e8\tau_{B(\theta^\star)}^3\norm{P(\theta^\star)}^{14}\norm{\hat\theta-\theta}^3 \\
        &\leq \paren{32\tau^2_{B(\theta^\star)}+ 2.5e3\tau^3_{B(\theta^\star)}}\norm{P(\theta^\star)}^4 \norm{\Gamma} 
    \end{align*}
    from \Cref{lem: helper lemma for RC} and the diameter condition. 
    It follows from \Cref{lem: lyap perturbation} that
    \begin{align*}
        &\bfE_{\theta\sim\calD}\norm{\Sigma^{K_{\mathsf{DR}}(\calD)}(\theta^\star) - \Sigma^{K_{\mathsf{DR}}(\calD)}(\theta)} \\
        &\leq \sup_{\theta\in G}\norm{\Sigma^{K_{\mathsf{DR}}(\calD)}(\theta^\star)}^{3/2}\norm{\Sigma^{K_{\mathsf{DR}}(\calD)}(\theta)}\norm{\theta-\theta^\star}\paren{1+\norm{K_{\mathsf{DR}}(\theta)}}\paren{2+\norm{\theta-\theta^\star}\paren{1+\norm{K_{\mathsf{DR}}(\theta)}}}.
    \end{align*}
    Let $\zeta(\theta) = \norm{\theta-\theta^\star}\paren{1+\norm{K_{\mathsf{DR}}(\theta)}}$, then we get the fifth fact: 
    \begin{align*}
        &\norm{\Sigma^{K_{\mathsf{DR}}(\calD)(\theta)}} 
        \leq \sup_{\theta\in G}\norm{\Sigma^{K_{\mathsf{DR}}(\calD)(\theta^\star)}} + \norm{\Sigma^{K_{\mathsf{DR}}(\calD)}(\theta^\star)}^{3/2}\norm{\Sigma^{K_{\mathsf{DR}}(\calD)}(\theta)}\zeta(\theta)(2+\zeta(\theta)) \\
        &\iff \norm{\Sigma^{K_{\mathsf{DR}}(\calD)}(\theta)} 
        \leq \sup_{\theta\in G}\frac{\norm{\Sigma^{K_{\mathsf{DR}}(\calD)}(\theta^\star)}}{1-\norm{\Sigma^{K_{\mathsf{DR}}}(\calD)(\theta^\star)}^{3/2}\zeta(\theta)(2+\zeta(\theta))} \leq 2\norm{\Sigma^{K_{\mathsf{DR}}(\calD)}(\theta^\star)}, 
    \end{align*}
    where the last inequality holds as long as $\norm{\Gamma}\leq\frac{1}{8\norm{\Sigma^{K_{\mathsf{DR}}}(\calD)(\theta^\star)}^{3/2}\paren{1+\norm{K_{\mathsf{DR}}}}^2}$. 

    Similarly, from \Cref{lem: lyap perturbation}, it holds that
    \begin{align*}
        &\norm{\Sigma^{K_{\mathsf{DR}}(\calD)}(\theta^\star) - \Sigma^{K(\theta^\star)}(\theta^\star)} \\
        &\leq \norm{\Sigma^{K(\theta^\star)}(\theta^\star)}^{3/2}\norm{\Sigma^{K_{\mathsf{DR}}(\calD)}(\theta^\star)}\norm{B(\theta^\star)\paren{K_{\mathsf{DR}}(\calD) - K(\theta^\star)}}\paren{2 + \norm{B(\theta^\star)\paren{K_{\mathsf{DR}}(\calD) - K(\theta^\star)}}} 
    \end{align*}  
    Let $\eta = \norm{B(\theta^\star)\paren{K_{\mathsf{DR}}(\calD) - K(\theta^\star)}}$, Then
    \begin{align*}
        & \norm{\Sigma^{K_{\mathsf{DR}}(\calD)}(\theta^\star)}\leq  \norm{\Sigma^{K(\theta^\star)}(\theta^\star)}+ \norm{\Sigma^{K(\theta^\star)}(\theta^\star)}^{3/2}\norm{\Sigma^{K_{\mathsf{DR}}(\calD)}(\theta^\star)}\eta(2+\eta) \\
        &\iff \norm{\Sigma^{K_{\mathsf{DR}}(\calD)}(\theta^\star)} \leq \frac{\norm{\Sigma^{K(\theta^\star)}(\theta^\star)}}{1-\norm{\Sigma^{K(\theta^\star)}(\theta^\star)}^{3/2}\eta(2+\eta)} \leq 2\norm{\Sigma^{K_{\mathsf{DR}}(\calD)}(\theta^\star)} 
    \end{align*}
    where the last inequality holds when
    \begin{align*}
        \eta(2+\eta) \leq \frac{1}{2\norm{\Sigma^{K(\theta^\star)}(\theta^\star)}^{3/2}}
    \end{align*}
    Now from the second and fourth fact, \eqref{eq: DR objective upper bound with theta} and \Cref{lem: helper lemma for RC}, $\eta$ is bounded as
    \begin{align*}
        &\eta = \norm{B(\theta^\star)\paren{K_{\mathsf{DR}}(\calD) - K(\theta^\star)}} \\
        &\leq \sup_{\theta\in G}\tau_{B(\theta^\star)}\paren{\textnormal{DR Objective} + 32 \norm{P(\theta^\star)}^{7/2} \norm{\theta-\theta^\star}} \\
        &\leq 2.5e3\tau^3_{B(\theta^\star)}\norm{P(\theta^\star)}^4 \norm{\Gamma} + 32\tau_{B(\theta^\star)}\norm{P(\theta^\star)}^{7/2}\norm{\theta-\theta^\star} = 2.6e3\tau^4_{B(\theta^\star)}\norm{P(\theta^\star)}^4\norm{\Gamma} \\
        &\iff \eta^2+2\eta \leq 3.2e4\tau^8_{B(\theta^\star)}\norm{P(\theta^\star)}^4\norm{\Gamma}
    \end{align*}
    where we applied $\mathsf{diam(G)}\leq\frac{1}{256}\inf_{\theta\in G}\norm{P(\theta)}^{-5}$ multiple times.
    Therefore the sixth inequality holds as long as
    \begin{align*}
        \norm{\Gamma} \leq \frac{1}{6.4e4\tau^8_{B(\theta^\star)}\norm{P(\theta^\star)}^4\norm{\Sigma^{K(\theta^\star)}(\theta^\star)}^{3/2}}
    \end{align*}
\end{proof}

\begin{lemma}
    \label{lem: DR suboptimality upper bound}
    Let $G$ denote the ellipsoid:
    \begin{align*}
        G \triangleq \{\theta ~:~ \theta=\hat \theta + \Gamma w, ~ w\in\mathcal{B}(0, 1)\},
    \end{align*} 
    Suppose the diameter of $G$ satisfies $\mathsf{diam}(G)\leq\frac{1}{6.4e4}\inf_{\theta\in G}\norm{P(\theta)}^{-5.5} \tau_{B(\theta^\star)}^{-8}$ and $\theta^\star\in G$. 
    Let $\calD$ be a distribution of system parameters over $G$.  Then it holds that
     \begin{align*}
        &C(K_{\mathsf{DR}}(\calD), \theta^\star) - C(K(\theta^\star), \theta^\star) \leq 2\norm{\hat\theta-\theta^\star}^2_{H(\theta^\star)} + 4\trace\paren{\mathbf{V}(\calD) H(\theta^\star)}+ L_{\mathsf{DR}}(\theta^\star)\norm{\Gamma}^3
    \end{align*}
    where
    \begin{align*}
        L_{\mathsf{DR}}(\theta^\star) = 5e8\du\tau_{B(\theta^\star)}^8\norm{P(\theta^\star)}^{17}
    \end{align*}
\end{lemma}

\begin{proof}
    By the fact that $K_{\mathsf{DR}}(\calD)$ stabilizes the system $\theta^\star$, we may write by \Cref{lem: performance difference}
    \begin{equation}
    \begin{aligned}
        \label{eq: DR excess cost decomposiotion}
        &C(K_{\mathsf{DR}}(\calD), \theta^\star) - C(K(\theta^\star), \theta^\star) = \trace\paren{(K_{\mathsf{DR}}(\calD) - K(\theta^\star)) \Sigma^{K_{\mathsf{DR}}(\calD)}(\theta^\star)  (K_{\mathsf{DR}}(\calD) - K(\theta^\star)) \Psi(\theta^\star)} \\
        &\leq 2 \bfE_{\theta\sim\calD} \brac{\trace\paren{(K_{\mathsf{DR}}(\calD) - K(\theta)) \Sigma^{K_{\mathsf{DR}}(\calD)}(\theta^\star)  (K_{\mathsf{DR}}(\calD) - K(\theta)) \Psi(\theta^\star)}} \\
        &+ 2 \bfE_{\theta\sim\calD} \brac{\trace\paren{(K(\theta) - K(\theta^\star)) \Sigma^{K_{\mathsf{DR}}(\calD)}(\theta^\star)  (K(\theta) - K(\theta^\star)) \Psi(\theta^\star)}},
    \end{aligned}
    \end{equation}
    where the inequality follows from the Cauchy-Schwarz. Next, we massage the first term to have the form of the domain randomization objective, and the second term to have the form of the certainty equivalence objective. To this end, first note that by \Cref{lem: helper lemma for RC} and $\Psi(\theta)\succeq I$, we get
    \begin{align*}
        \Psi(\theta^\star) \preceq \Psi(\theta)\paren{1 + 15 \tau_{B(\theta^\star)}^2 \norm{P(\theta^\star)}^3 \norm{\theta - \theta^\star}}.
    \end{align*}
    Additionally, write expand 
    \begin{align*}
       \Sigma^{K_{\mathsf{DR}}(\calD)}(\theta^\star) &\preceq \Sigma^{K_{\mathsf{DR}}(\calD)}(\theta) + I \times \norm{\Sigma^{K_{\mathsf{DR}}(\calD)}(\theta^\star) - \Sigma^{K_{\mathsf{DR}}(\calD)}(\theta)}  \mbox{ and }\\
       \Sigma^{K_{\mathsf{DR}}(\calD)}(\theta^\star) &= \Sigma^{K(\theta^\star)}(\theta^\star) + I \times \norm{\Sigma^{K_{\mathsf{DR}}(\calD)}(\theta^\star) - \Sigma^{K(\theta^\star)}(\theta^\star)}.
    \end{align*}
    Substituting the above set of inequalities into \eqref{eq: DR excess cost decomposiotion} we achieve the bound 
    \begin{align}
        &C(K_{\mathsf{DR}}(\calD), \theta^\star) - C(K(\theta^\star), \theta^\star) \nonumber \\
        &\leq 2 \times  \mbox{DR objective}\times \paren{1 + 15 \tau_{B(\theta^\star)}^2 \norm{P(\theta^\star)}^3 \norm{\theta - \theta^\star}} \nonumber\\
        &+ 2 \times \mbox{Expected CE cost gap} \nonumber\\
        & + 2 \du \norm{\Psi(\theta^\star)} \bfE_{\theta\sim\calD}\brac{\norm{\Sigma^{K_{\mathsf{DR}}(\calD)}(\theta^\star) - \Sigma^{K_{\mathsf{DR}}(\calD)}(\theta)} \norm{K_{\mathsf{DR}}(\calD) - K(\theta)}^2} \nonumber\\
        &+ 2 \du \norm{\Psi(\theta^\star)} \norm{\Sigma^{K_{\mathsf{DR}}(\calD)}(\theta^\star) - \Sigma^{K(\theta^\star)}(\theta^\star)}\bfE_{\theta\sim\calD}\brac{ \norm{K(\theta^\star) - K(\theta)}^2}, \label{eq:DR suboptimality gap}
    \end{align}
    where 
    \begin{align*}
        &\mbox{Expected CE cost gap}  = \bfE_{\theta\sim\calD} \brac{\trace\paren{(K(\theta) - K(\theta^\star)) \Sigma^{K(\theta^\star)}(\theta^\star)  (K(\theta) - K(\theta^\star)) \Psi(\theta^\star)}}, \\
         &\mbox{DR Objective }=  \bfE_{\theta\sim\calD} \brac{\trace\paren{(K_{\mathsf{DR}}(\calD) - K(\theta)) \Sigma^{K_{\mathsf{DR}}(\calD)}(\theta)  (K_{\mathsf{DR}}(\calD) - K(\theta)) \Psi(\theta)}}.
    \end{align*}
    We let $D(\theta)$ denote the DR objective.
    From \Cref{lem: DR objective upper bound}, the first term of \eqref{eq:DR suboptimality gap} becomes
    \begin{align}
        &2\times D(\theta)\times \paren{1 + 15 \tau_{B(\theta^\star)}^2 \norm{P(\theta^\star)}^3 \norm{\theta - \theta^\star}}\nonumber\\
        &\leq 2 \times \paren{\trace\paren{\mathbf{V}(\calD) H(\theta^\star)} + L_{\mathsf{DR,1}}(\theta^\star)  \norm{\Gamma}^3} \times \paren{1 + 15 \tau_{B(\theta^\star)}^2 \norm{P(\theta^\star)}^3\norm{\theta - \theta^\star}} \nonumber\\
        &\leq 2\trace\paren{\mathbf{V}(\calD) H(\theta^\star)} + 3.5e8\tau_{B(\theta^\star)}^5\norm{P(\theta^\star)}^{17}\norm{\Gamma}^3 \label{eq: first term of DR suboptimality gap}
    \end{align}
    where we used \Cref{lem: helper lemma for RC} and applied the closeness condition to $\norm{\theta-\theta^\star}$ in the last inequality. 
    Next the second term of \eqref{eq:DR suboptimality gap} follows from \Cref{lem: CE upper bound} that
    \begin{align}
        2 \times \mbox{Expected CE cost gap} \leq \bfE_{\theta\sim\calD}\brac{2\norm{\theta-\theta^\star}^2_{H(\theta^\star)}} + 1.2e6\tau_{B(\theta^\star)}^3\norm{P(\theta^\star)}^{14}\norm{\Gamma}^3 \label{eq:second term of DR suboptimality gap}
    \end{align}
    Then from the fifth fact of \Cref{lem: DR helper lemmas}, 
    \begin{align*}
        &\bfE_{\theta\sim\calD}\norm{\Sigma^{K_{\mathsf{DR}}(\calD)}(\theta^\star) - \Sigma^{K_{\mathsf{DR}}(\calD)}(\theta)} \\
        &\leq\sup_{\theta\in G}3\norm{\Sigma^{K(\theta^\star)}(\theta^\star)}^{5/2}\paren{1+\norm{K_{\mathsf{DR}}(\calD)}}\norm{\theta-\theta^\star}\paren{2 + \frac{1}{8\norm{\Sigma^{K_{\mathsf{DR}}}(\calD)(\theta^\star)}^{3/2}\paren{1+\norm{K_{\mathsf{DR}}(\calD)}}}}
        \\
        &\leq \sup_{\theta\in G}5\norm{\Sigma^{K(\theta^\star)}(\theta^\star)}^{5/2}\paren{1+\norm{K_{\mathsf{DR}}(\calD)}}\norm{\theta-\theta^\star}
    \end{align*}
    since $\norm{\Sigma^{K(\theta^\star)}(\theta^\star)} \geq 1$. Also from the sixth fact of \Cref{lem: DR helper lemmas}
    \begin{align*}
        &\norm{\Sigma^{K_{\mathsf{DR}}(\calD)}(\theta^\star) - \Sigma^{K(\theta^\star)}(\theta^\star)} \leq \sup_{\theta\in G}6.4e4\norm{P(\theta^\star)}^{13/2}\norm{\theta-\theta^\star}
    \end{align*}
    By combining with \Cref{lem: Riccati perturbation}, the third term of \eqref{eq:DR suboptimality gap} becomes
    \begin{align}
        &2 \du \norm{\Psi(\theta^\star)} \bfE_{\theta\sim\calD}\brac{\norm{\Sigma^{K_{\mathsf{DR}}(\calD)}(\theta^\star) - \Sigma^{K_{\mathsf{DR}}(\calD)}(\theta)} \norm{K_{\mathsf{DR}}(\calD) - K(\theta)}^2} \nonumber \\
        &\leq \du \norm{\Psi(\theta^\star)}  \paren{8e4\tau_{B(\theta^\star)}^2\paren{1+\norm{K_{\mathsf{DR}}(\theta)}}\norm{P(\theta^\star)}^{23/2} + 3.4e8\tau_{B(\theta^\star)}^3\norm{P(\theta^\star)}^{39/2}\norm{\Gamma}}\norm{\Gamma}^3 \nonumber  \\
        &\leq 5e6\du\tau^8_{B(\theta)}\norm{P(\theta^\star)}^{33/2}\norm{\Gamma}^3 \label{eq:third term of DR suboptimality gap}
    \end{align}
    where the last inequality follows from \Cref{lem: simplifying inequalities} and noting that $1+\norm{K_{\mathsf{DR}}(\theta)}\leq12\tau_{B(\theta^\star)}^4\norm{P(\theta^\star)}^{1/2}$ from \Cref{lem: DR helper lemmas}.
    Furthermore, the fourth term of \eqref{eq:DR suboptimality gap} is given by
    \begin{align}
        &2 \du \norm{\Psi(\theta^\star)} \norm{\Sigma^{K_{\mathsf{DR}}(\calD)}(\theta^\star) - \Sigma^{K(\theta^\star)}(\theta^\star)}\bfE_{\theta\sim\calD}\brac{ \norm{K(\theta^\star) - K(\theta)}^2} \leq 1.4e8 \du \tau^2_{B(\theta)}\norm{P(\theta^\star)}^{31/2}\norm{\Gamma}^3 \label{eq:fourth term of DR suboptimality gap}
    \end{align}
    where the last inequatliy follows from \Cref{lem: Riccati perturbation} and \Cref{lem: simplifying inequalities}. 
    Finally, from \eqref{eq: first term of DR suboptimality gap}, \eqref{eq:second term of DR suboptimality gap}, \eqref{eq:third term of DR suboptimality gap} and \eqref{eq:fourth term of DR suboptimality gap}, we get
    \begin{align*}
         &C(K_{\mathsf{DR}}(\calD), \theta^\star) - C(K(\theta^\star), \theta^\star) \\
         &\leq 2\bfE_{\theta\sim\calD}{\norm{\theta-\theta^\star}^2_{H(\theta^\star)}} + 2\trace\paren{\mathbf{V}(\calD) H(\theta^\star)} + L_{\mathsf{DR}}(\theta^\star)\norm{\Gamma}^3
    \end{align*}
    where
    \begin{align*}
        L_{\mathsf{DR}}(\theta^\star) &= \paren{3.5e8\tau_{B(\theta^\star)}^5 + 1.2e6\tau_{B(\theta^\star)}^3 + 5e6\du\tau^8_{B(\theta)} +  1.4e8\du \tau^2_{B(\theta)}}\norm{P(\theta^\star)}^{17}\\
        &= 5e8\du\tau_{B(\theta^\star)}^8\norm{P(\theta^\star)}^{17}
    \end{align*}
    We conclude the proof by noting that
    \begin{align*}
        2\bfE_{\theta\sim\calD}\norm{\theta-\theta^\star}^2_{H(\theta^\star)} = 
        2\norm{\hat\theta - \theta^\star}_{H(\theta^\star)}^2 + 
        2\trace\paren{\mathbf{V}(\calD) H(\theta^\star)}.
    \end{align*}
\end{proof}

\subsection{Proof of \texorpdfstring{\Cref{lem: domain randomization general}}{}}
\label{subsec: Proof of DR upper bound}
\begin{proof}
    As proved in \Cref{subsec: proof of RC upper bound}, $\theta^\star\in G$ with probability at least $1-\delta$. 
    From the diameter condition in \Cref{lem: DR suboptimality upper bound}, i.e., $\mathsf{diam(G)}\leq\frac{1}{6.4e4}\inf_{\theta\in G}\norm{P(\theta)}^{-5.5} \tau_{B(\theta^\star)}^{-8}$, $N$ must satisfy
    \begin{align*}
        & \sup_{\theta_1, \theta_2\in G}\norm{\theta_1-\theta_2}^2\leq \frac{32\paren{d_\theta+\log\frac{2}{\delta}}}{N\lambda_{\min}\paren{\hat{\mathsf{FI}}}} \leq \frac{1}{(6.4e4)^2}\inf_{\theta\in G}\norm{P(\theta)}^{-11} \tau_{B(\theta^\star)}^{-16}.  \\
        &\iff N \geq \frac{5e6\norm{P(\theta^\star)}^{11} \tau_{B(\theta^\star)}^{16}\paren{d_\theta+\log\frac{2}{\delta}}}{\lambda_{\min}\paren{\mathsf{FI}(\theta^\star)}}
    \end{align*} 
    where we applied $0.5\mathsf{FI}(\theta^\star)\preceq\hat{\mathsf{FI}}$ and \Cref{lem: Riccati perturbation} to get the last inequality.
    As long as this holds, we get \eqref{eq:DR upper bound} from \Cref{thm: identification bound} and \Cref{lem: DR suboptimality upper bound}. 
\end{proof}

\subsection{Proof of \texorpdfstring{\Cref{thm: Domain Randomization Upper Bound}}{}}
\label{subsec: Proof of DR upper bound w/ uniform distribution}
\begin{proof}
    Since the variance of uniform distribution over the ellipsoid $\curly{w \colon w^T\Gamma w\leq1, w\in\R^d}$ is given as $\frac{1}{d+2}\Gamma^{-1}$, 
    \begin{align*}
        \trace(V(\calD)H(\theta^\star)) &\leq \frac{16\paren{d_\theta+\log\frac{2}{\delta}}}{Nd_\theta}\trace\paren{\hat{\mathsf{FI}}^{-1}H(\theta^\star)} \leq \frac{32\paren{1+\log\frac{2}{\delta}/d_\theta}}{N}\trace\paren{{\mathsf{FI}}(\theta^\star)^{-1}H(\theta^\star})
    \end{align*}
    where we used $0.5\hat{\mathsf{FI}}\preceq {\mathsf{FI}}(\theta^\star)$ in the last inequality. Then from \Cref{lem: domain randomization general}, we get \eqref{eq:DR upper bound w/ unif dist}. 
\end{proof}

\section{Implementation Details}
\label{s: implementation details}
\subsection{Linear System}
We explain the details about the linear system experiments in \Cref{s: numerical}. The goal is to compare the three controller synthesis methods considered in this work (\Cref{s: methods}) to study show how the suboptimality gap of the LQR cost changes. 
All the code is implemented in JAX \citep{jax2018github}. 
First we identify the least square estimate $\hat\theta$ by solving \eqref{eq: least squares} given collected data by running a variable number of experiments with $U_t\sim\mathcal{N}(0, I)$ and $W_t\sim\mathcal{N}(0, I)$, and followed by the controller design. 
\begin{itemize}
\item 
For the certainty equivalence controller synthesis, we simply synthesize a controller by solving the \textit{Algebraic Ricatti Equation}. 
\item 
For the domain randomization controller synthesis, we implement \Cref{alg: dr lqr} via stochastic gradient descent. 
We choose $\mathcal{D}$ as a uniform distribution over the confidence ellipsoid $G$. We set $M = 10000$ and $\eta=0.0005$ in \Cref{alg: dr lqr}. We use the closed-form expression of the policy gradient as described in Lemma $1$ of \cite{fazel2018global} to implement gradient descent. 
\item 
For the robust controller synthesis, we adopt the scenario-approach \citep{calafiore2006scenario}, where we sample 30 points from the ellipsoid $G$ \eqref{eq: confidence ellipsoid}, by formulating as the semidefinite programming problem with linear matrix inequalities \citep{caverly2019lmi} and solve the convex optimization problem using CVXPY \citep{diamond2016cvxpy} using the SCS solver \citep{ocpb:16}. 
\end{itemize}

Since these randomization procedures result in high variance, we perform $500$ trials and report the median and quartiles over those trials in \Cref{fig:result dr lqr}. 

\subsection{Pendulum}
The pendulum is goverend by the following dynamics
\begin{align*}
    X_{t+1} = f(X_t, U_t+W_t; \theta^\star),
\end{align*}
where the state space $X_t\in\R^2$ consists of angle $\psi_t$(rad) and angular velocity $\dot\psi_t$(rad/s), and the action $U_t\in\R$ is the actuation torque $\tau_t$(N m) applied directly to the free end of the pendulum, which is corrupted by i.i.d Gaussian noise $W_t \sim \mathcal{N}(0, 1)$. The unknown parameter $\theta\in\R^3$ accounts for the mass $m$ and the length $l$ of the pole, and the gravitiy $g$, which take values $1.0$, $1.0$, and $9.81$, respectively.  

As in the linear experiments, first we identify the estimate $\hat\theta$ by solving the following least squares problem:
\begin{align*}
    \hat\theta \triangleq \argmin_\theta \sum_{n=1}^N\sum_{t=1}^T\norm{X_{t+1}^n - f(X_t, U_t; \theta)}^2.
\end{align*}
The data consists of varying numbers of length $10$ trajectories collected from the pendulum starting from the downward position ($\theta = \pi, \dot \theta =0$) with $U_t\sim\mathcal{N}(0, 1)$ and $W_t\sim\mathcal{N}(0, 1)$ 
This estimation procedure is followed by the sampling-based Model Predictive Controller synthesis based on the cross-entropy method \citep{botev2013cross} with the different design methods.
\begin{itemize}
    \item For the certainty equivalence control, actions are selected to  minimize the cost of the rollout in the receding horizon which is simulated with $\hat\theta$ as if it were the ground truth.
    \item For the domain randomization control, we first design a sampling distribution $\calD$ as a uniform distribution over a sphere centered at the estimate $\hat \theta$, with radius given by a hyperparameter $2.0$ divided by the number of trajectories used to identify $\theta$. We then sample $K=15$ instances $\theta_1, \dots, \theta_{K} \in\mathcal{D}$. The controller then plans sequences of actions which minimize the average cost of all the rollouts simulated with $\theta_1, \dots, \theta_{K}$. 
\end{itemize}
We set the stage cost $C(X_t, U_t)$ as
\begin{align*}
    C(X_t, U_t) = \left\{\begin{array}{ll}
        \psi_t^2 + 0.1\dot\psi_t^2 + 2.0\tau_t^2 & -\pi/4 \leq \psi \leq \pi/4 \\
        50.0 & \text{otherwise}
    \end{array}\right.
\end{align*}
This cost design promotes the controller to keep the pendulum upright using the minimal control effort. CE with an inaccurate parameter estimate might underestimate control effort required to avoid letting the pendulum fall.
On the other hand, DR plans such that it keeps the pendulum upright for all sampled systems $\theta_1, \dots, \theta_K$, therefore is able to keep upright position even when the estimates are poor.

For other hyperparameters and detailed implementation, refer to the code \footnote{Codes can be found at \url{https://github.com/Tesshuuuu/domain-randomization-l4dc2025}}.

\end{document}